\documentclass{amsart}

 \usepackage[foot]{amsaddr}
\usepackage{amsmath, bbm}
\usepackage{mathabx}
\usepackage{fullpage}
\usepackage{amssymb}
\usepackage{epstopdf}
\usepackage{algorithmic}
\usepackage{mathrsfs}
\usepackage{bm}
\usepackage{comment}
\usepackage{mathtools}
\usepackage{color}
\usepackage[
pdftitle={Flow-based VMC},
pdfcreator={pdflatex},
pdfsubject={preprint},
hyperindex = {true},
colorlinks = {true},
linkcolor = {blue},
citecolor = {blue}
]{hyperref}

\newcommand\sbullet[1][.5]{\mathbin{\vcenter{\hbox{\scalebox{#1}{$\bullet$}}}}}
\theoremstyle{definition}

\newtheorem{lem}{Lemma}[section]

\DeclareMathOperator{\Diff}{Diff}

\DeclareMathOperator*{\argmin}{arg\,min}
\DeclareMathOperator*{\argmax}{arg\,max}
\DeclareMathOperator*{\var}{var}
\DeclareMathOperator*{\cov}{cov}
\DeclareMathOperator*{\tr}{tr}
\DeclareMathOperator*{\diag}{diag}

\newcommand{\Norm}[2][]{\lVert#2\rVert\ifthenelse{\isempty{#1}}{}{_{#1}}}
\newcommand{\BigNorm}[2][]{\left\lVert#2\right\rVert\ifthenelse{\isempty{#1}}{}{_{#1}}}
\newcommand{\InProd}[3][]{\langle#2,#3\rangle\ifthenelse{\isempty{#1}}{}{_{#1}}}
\newcommand{\BigInProd}[3][]{\left\langle#2,#3\right\rangle\ifthenelse{\isempty{#1}}{}{_{#1}}}

\begin{document}

\author{James Stokes$^{1}$}
\address{$^1$Center for Computational Quantum Physics and Center for Computational Mathematics, Flatiron Institute, New York, NY 10010 USA}
\author{Brian Chen$^{2}$}
\author{Shravan Veerapaneni$^2$}
\address{$^2$Department of Mathematics, University of Michigan, Ann Arbor, MI 48109 USA}

\title{Numerical and geometrical aspects of flow-based variational quantum Monte Carlo}
\maketitle
\begin{abstract}
This article aims to summarize recent and ongoing efforts to simulate continuous-variable quantum systems using flow-based variational quantum Monte Carlo techniques, focusing for pedagogical purposes on the example of bosons in the field amplitude (quadrature) basis. Particular emphasis is placed on the variational real- and imaginary-time evolution problems, carefully reviewing the stochastic estimation of the time-dependent variational principles and their relationship with information geometry. 
Some practical instructions are provided to guide the implementation of a PyTorch code. 
The review is intended to be accessible to researchers interested in machine learning and quantum information science.
\end{abstract}

\tableofcontents

\section{Introduction} 
 Recent years have witnessed a profitable interchange between the fields of machine learning, quantum many-body physics and quantum information science. This multidisciplinary interaction has been partially facilitated by the discovery that artificial neural networks provide a powerful inductive bias for parametrizing subsets of quantum many-body Hilbert space. Although the description of Hilbert-space vectors via neural networks renders exact linear algebra operations prohibitive for this subset of quantum states, the existence of efficient stochastic approximation algorithms \cite{mcmillan1965ground, carleo2017solving} called variational Monte Carlo (VMC) has enabled neural-network-based quantum states (NQS) to accurately reveal properties of the ground state for quantum spin systems, as well as to simulate their time evolution using a time-dependent variant of VMC (the so-called t-VMC) \cite{carleo2012localization, carleo2014light}. Since the inception of complex-valued restricted Boltzmann machines \cite{carleo2017solving}, the reach of neural-network quantum states has grown to encompass a diversity of quantum systems,  made possible by the use of increasingly sophisticated (often multi-layered) architectures. Another driver of interaction is the discovery of close analogies between VMC and variational quantum algorithms (VQAs). Recent work on quantum information geometry in particular, \cite{stokes2020quantum} has clarified the connection between the natural gradient descent in machine learning \cite{amari1998natural}, stochastic reconfiguration VMC \cite{sorella2007weak} and variational imaginary-time evolution in quantum computing \cite{yuan2019theory}. 

This tutorial paper is intended to serve as a self-contained review of flow-based VMC and t-VMC for continuous-variable quantum systems. For the purpose of concreteness we frame the discussion around the example of bosonic quantum systems, represented in the field amplitude basis. The field amplitude basis has not been a traditional focal point of the VMC literature\footnote{See however \cite{han2020deep,han2021neural,rindaldi2021,stokes2021continuous}.}, which has concentrated on non-relativistic systems that are more interpretable in the Fock basis. The field amplitude basis is natural, however, in systems with relativistic symmetry, where the regulated bosonic Hamiltonian is represented on an $L^2$-space as a simple Schr\"{o}dinger operator. The simplicity of the Hamiltonian therefore also offers pedagogical advantages. A possible computational advantage of the field amplitude basis is that it does not require artificially restricting the allowed mode occupation numbers to finite range for its numerical implementation.  In an effort to foster further interactions between machine learning and variational quantum simulation, we chose to advance a geometrical approach based on information geometry. This contrasts with previous work on variational quantum simulation which focuses on K\"{a}hler geometry (see \cite{hackl2020geometry} for a review).

        In a related parallel line of work, flow-based probabilistic neural networks have been applied to investigate ground state properties of lattice-regulated Euclidean quantum field theories \cite{PhysRevLett.125.121601}. This approach hinges on the existence of a quantum-classical mapping, under which properties of the ground state are related to classical expectation values with respect to a Gibbs distribution with an intractable normalizing factor. The aim of this approach is to accelerate, or eliminate, an otherwise expensive Markov chain Monte Carlo sampler by approximating the Gibbs measure using normalizing flows. The work described in the following aims to expand the applicability of flow-based methods to situations in which the quantum-classical mapping does not produce an efficiently-scaling sampling algorithm, which is understood to be a symptom of a so-called sign/phase problem.
        
        It is illuminating to compare and contrast the approach advocated here with proposals in the quantum computing literature for quantum field simulation. Like \cite{jordan2014quantum}, we adopt the field amplitude basis, although we remain agnostic about the form of the regulator since it does not alter the general structure of the resulting Schr\"{o}dinger operator (see \cite{liu2021towards} for a discussion of possibilities). The choice to focus on normalizing flows is motivated by their exact sampling properties and potential to exploit distributed computing using the embarrassing parallelism of Monte Carlo \cite{pfau2020integrable, zhao2021overcoming}.
        
        In an effort to make the paper accessible to a wide audience we have chosen an exposition emphasizing mathematical and numerical aspects, minimizing where possible, the necessary physics prerequisites. 
        Although a fully rigorous treatment is possible, mathematical rigor is not attempted in this work. Detailed appendices have been included containing proofs of all relevant claims as well as additional physics motivation. Although our target application is bosons, this paper can be read independently as a self-contained review of VMC for continuous-variable quantum systems.
The paper is organized as follows. In section \ref{sec:preliminaries} we provide a brief refresher on normalizing flows, dynamics of isolated quantum systems and time-dependent variational principles. Section \ref{sec:informationgeometry} is devoted to explaining the geometry of quantum states and of classical probability densities (information geometry). Time-dependent variational principles are then reviewed in this context in section \ref{sec:TDVPs}, explaining the role of normalization and holomorphy. The next section specializes to the stochastic approximation of the time-dependent variational principles using VMC and t-VMC, including a discussion of variance reduction. Section \ref{sec:comments} comments on prospects and challenges for modeling quantum states with symmetries using equivariant flows. In section \ref{sec:experiments} we sketch a worked example in PyTorch \cite{pytorch} using normalizing flows to prepare bosonic ground eigenfunctions, showing improved energy compared to optimized Gaussian states. Section \ref{sec:conclusions} concludes with a discussion of future directions.

\section{Notation}\label{sec:notation}
Let $L^p(\mathbb{R}^d) := L^p(\mathbb{R}^d;\mathbb{C})$ denote the normed space of complex-valued, $p$-integrable functions with $p \geq 1$ and norm
$
\Vert \cdot \Vert_p
$.
Denote by $\langle \cdot | \cdot \rangle$ the $L^2$ inner product, with the convention chosen to be linear in the second argument (anti-linear in the first argument). Denote by $\Diff(\mathbb{R}^d)$ the group of smooth bijections from $\mathbb{R}^d$ to itself with smooth inverse and binary operation given by composition. For random vectors $X$ and $Y$, valued in $\mathbb{C}^m$ and $\mathbb{C}^n$, respectively define the covariance matrix by
$
    \cov(X,Y) := \mathbb{E}[X Y^\dag] - \mathbb{E}[X] \mathbb{E}[Y]^\dag
    \in \mathbb{C}^{m\times n}
$, where $\dag$ denotes the conjugate transpose. 
Denote by $\mathbb{S}^n$, $\mathbb{S}^n_+$ and $\mathbb{S}^n_{++}$, respectively the real symmetric, symmetric positive semi-definite and symmetric positive definite matrices and similarly denote the corresponding complex Hermitian matrices as $\mathbb{H}^n,\mathbb{H}^n_+$ and $\mathbb{H}^n_{++}$. The identity matrix is denoted by $\mathbbm{1}$.

\section{Preliminaries}\label{sec:preliminaries}

Although a formulation of normalizing flows applicable to general Riemannian manifolds is possible
\cite{rezende2020normalizing}\footnote{Indeed necessary for the description of nonlinear sigma-model field theories.}, the exposition is considerably simplified by restricting attention to flows defined on Euclidean spaces. 
The starting point is to recall that the diffeomorphism group $\Diff(\mathbb{R}^d)$ acts naturally on the normed space $L^p(\mathbb{R}^d)$. In particular, a function $\psi \in L^p(\mathbb{R}^d)$ transforms under a diffeomorphism $f \in \Diff(\mathbb{R}^d)$ to $f\cdot\psi \in L^p(\mathbb{R}^d)$ defined for all $x\in\mathbb{R}^d$ by,
\begin{equation}\label{e:trans}
    (f\cdot \psi)(x) := \left|\det J_{f^{-1}}(x)\right|^{1/p} \psi\big(f^{-1}(x)\big) \enspace ,
\end{equation}
in such a way that the $p$-norm is preserved,
\begin{equation}\label{e:pnorm}
    \Vert f \cdot \psi \Vert_p = \Vert \psi \Vert_p \enspace .
\end{equation}
The action of $\Diff(\mathbb{R}^d)$ is compatible with the group operation in the sense that for  all $f,g \in \Diff(\mathbb{R}^d)$,
\begin{equation}\label{e:grouppcompat}
    f \cdot (g \cdot \psi) = (f\circ g)\cdot \psi \enspace .
\end{equation} 

Specializing to $p=1$, a normalized, non-negative (probability) density $p \in L^1(\mathbb{R}^d)$ is carried under the action of $f$ to a normalized probability density $f\cdot p$. The above observation has motivated the investigation of normalizing flows for probabilistic modeling \cite{rezende2015variational, hackett2021flow}. The aim of this endeavor is approximate a complex, multi-modal target probability density, starting from a simple normalized probability density $p$ (such as a multi-variate Gaussian), by choosing the diffeomorphism amongst a parametrized differentiable family of invertible multi-layer neural networks. If $\theta \in \mathbb{R}^n$ denotes the parameters indexing a family of diffeomorphisms $f_\theta$, then density estimation using normalizing flows is to be performed under the hypothesis $\{ p_\theta : \theta \in \mathbb{R}^n\}$ where $p_\theta := f_\theta \cdot p$.

The successful application of normalizing flows to density estimation tasks has motivated their generalization to quantum applications. In particular, setting $p=2$, normalized densities are replaced by normalizable wavefunctions $\psi \in L^2(\mathbb{R}^d)$ subject to the norm constraint $\Vert \psi \Vert_2 = 1$ and the diffeomorphism group acts on the Hilbert space $L^2(\mathbb{R}^d)$ by unitary transformations \cite{cranmer2019inferring, xie2021ab}. In particular, the linear operator $U_f : L^2(\mathbb{R}^d) \to L^2(\mathbb{R}^d)$ defined for each diffeomorphism $f\in\Diff(\mathbb{R}^d)$ by $U_f : \psi \mapsto f\cdot \psi$ is unitary. This fact, which follows immediately from \eqref{e:pnorm} using the polarization identity, can also be seen as a consequence of the following identity for all $\psi,\psi' \in L^2(\mathbb{R}^d)$, \begin{equation}\label{e:unitarity}
    \left\langle \psi \,\middle|\, U_f (\psi') \right\rangle = \big\langle  U_f^{-1}(\psi) \, \big|\, \psi' \big\rangle \enspace .
\end{equation}
A differentiable family of unit-normalized $L^2$ functions $\psi_\theta := f_\theta \cdot \psi$ is obtained by choosing a base function $\psi \in L^2(\mathbb{R}^d)$ satisfying $\Vert \psi \Vert_2 = 1$. Following von Neumann, the space of states for a quantum system described via the Hilbert space $L^2(\mathbb{R}^d)$ is given by set of rank-1 orthogonal projection operators (projectors). The projector $P_\psi$ onto the complex line spanned by the unit vector $\psi \in L^2(\mathbb{R}^d)$ is a linear map defined for all $\psi' \in L^2(\mathbb{R}^d)$ by
$
    P_\psi : \psi' \mapsto \langle \psi | \psi'\rangle \, \psi
$.
As we shall review in the subsequent sections, this projector viewpoint of quantum states turns out to be the most natural one for discussing time-dependent variational principles. 

The subject of this paper is isolated quantum systems whose dynamics is dictated by a known time-independent Hermitian Hamiltonian operator $H$. If the system is initialized in the state $P_\psi$, then the uninterrupted time evolution of the system is described by the sequence of projectors $[0,\infty) \ni t \mapsto P_{\exp(-iHt)\psi}$. It is likewise instructive to consider the sequence of states $[0,\infty) \ni t \mapsto P_{\exp(-Ht)\psi}$ which correspond to unphysical evolution along the imaginary time axis. An important application of imaginary-time evolution is to preparation of a ground eigenfunction, since the component $\psi_\perp$ of $\psi$ lying orthogonal to the ground space experiences exponential damping relative to the parallel component $\psi_\parallel$.

Given initial parameters $\theta_0 \in \mathbb{R}^n$, a time-dependent variational principle is a proposed sequence of parameters $(\theta_t)_{t\geq0}$ such that $(P_{\psi_{\theta_t}})_{t\geq0}$ optimally describes the exact projector evolution $(P_{\exp(-At)\psi_0})_{t\geq0}$ where $A$ denotes either $H$ or $iH$, for imaginary- or real-time evolution, respectively. In practice, the sequence of parameters is defined implicitly by a system of ordinary differential equations, which must be approximated, leading to accumulation of error with evolution time, in addition to the systematic bias originating in the finite capacity of the variational family. In the case of imaginary time evolution, error accumulation is not of concern if the ultimate goal is preparation of a ground eigenfunction for a Hamiltonian with bounded-below energy, since the Rayleigh-Ritz principle ensures that any trial wavefunction provides an upper bound for the exact energy. Indeed energy optimization using neural-network-based quantum states is an active field of research and a number of proposals have been put forward including stochastic sign gradient descent \cite{luo2019backflow} and Gauss-Newton \cite{webber2021rayleigh}.

\section{Information geometry}\label{sec:informationgeometry}
Although we will ultimately be concerned with the geometry of quantum states, it is illuminating to first consider the geometry of classical probability densities. Indeed, these notions will be seen to coincide for a wide family of quantum states.
There is a natural distance metric on the set of normalized probability densities called the Fisher-Rao distance. Given probability densities $p,q \in L^1(\mathbb{R}^d)$ the Fisher-Rao distance is defined by
\begin{equation}
    d_{\rm FR}(p,q) := \arccos\big(\big\langle p^{1/2}\big|q^{1/2} \big\rangle \big)
\end{equation}
where $p^{1/2} \in L^2(\mathbb{R}^d)$ denotes the pointwise square root. The Fisher-Rao distance is manifestly invariant under arbitrary diffeomorphisms $f \in \Diff(\mathbb{R}^d)$,
\begin{equation}\label{e:FRdiffinvariance}
    d_{\rm FR}(p,q) = d_{\rm FR}(f \cdot p,f\cdot q) \enspace .
\end{equation}
In order to expose the Riemannian structure underling $d_{\rm FR}$, it is useful to consider a parametrized differentiable family of probability densities $\{p_\theta : \theta \in \mathbb{R}^n\}$. Define the symmetric positive semi-definite information matrix $I(\theta) \in \mathbb{S}^{n}_+$ for all $\theta \in \mathbb{R}^n$,
\begin{equation}\label{e:FIM}
    I(\theta) = I[p_\theta] := \underset{x \sim p_\theta}{\mathbb{E}}\left[\nabla_\theta \log p_\theta(x) \, \nabla_\theta \log p_\theta(x)^T\right] \enspace ,
\end{equation}
which is invariant under diffeomorphisms $f \in \Diff(\mathbb{R}^d)$,
\begin{equation}\label{e:FIMdiffinvariance}
    I[f\cdot p_\theta] = I[p_\theta] \enspace .
\end{equation}
In addition, the information matrix transforms as a covariant tensor under diffeomorphisms  of the parameter manifold. In particular, for a diffeomorphism $\phi \in \Diff(\mathbb{R}^n)$, define $\widetilde{p}_\theta := p_{\phi(\theta)}$ and then
\begin{equation}\label{e:covariance}
    \widetilde{I}(\theta):=I[\widetilde{p}_\theta] = J_\phi(\theta)^{T}I(\phi(\theta))J_\phi(\theta) \enspace .
\end{equation}
The above observations concerning the information matrix are in accord with the fact that it provides the coefficient matrix for the infinitesimal line element obtained by restricting the Fisher-Rao distance to the parametric family,
\begin{equation}
	d_{\rm FR}^2(p_\theta, p_{\theta + {\rm d}\theta}) = \frac{1}{4} \sum_{\mu,\nu=1}^n I_{\mu\nu}(\theta) {\rm d} \theta^{\mu} {\rm d}\theta^{\nu} \enspace .
\end{equation}
Although the information matrix is not a Riemannian metric tensor since it fails the requirement of non-degeneracy, it is however, the pull-back of a Riemannian metric tensor on the infinite-dimensional manifold of strictly positive probability densities\footnote{The so-called the Fisher-Rao metric tensor, which is uniquely characterized by the requirement of diffeomorphism invariance \cite{bauer2016uniqueness}. The Fisher-Rao distance is the geodesic distance function corresponding to the Fisher-Rao metric tensor.}.

In passing to quantum state
space, the natural distance function generalizes to Fubini-Study, which assigns a distance between the projectors onto unit vectors $\psi,\psi' \in L^2(\mathbb{R}^d)$ as follows,
\begin{equation}
    d_{\rm FS}(P_\psi,P_{\psi'}) := \arccos(|\langle \psi | \psi' \rangle |) \enspace .
\end{equation}
The Fubini-Study distance inherits the diffeomorphism invariance of Fisher-Rao, and expands it to all unitary transformations $U :L^2(\mathbb{R}^d) \to L^2(\mathbb{R}^d)$,
\begin{equation}
    d_{\rm FS} (P_\psi,P_{\psi'}) = d_{\rm FS}\big(P_{U(\psi)},P_{U(\psi')}\big) \enspace .
\end{equation}
Paralleling the discussion for probability densities above, now consider a differentiable family of unit-normalized $L^2$ functions $\{ \psi_\theta : \theta \in \mathbb{R}^n \}$ and define the quantum geometric tensor \cite{berry1989quantum} for all $\theta \in \mathbb{R}^n$ as the Hermitian positive semi-definite matrix  $G(\theta) \in \mathbb{H}^{n}_+$ with components,
\begin{equation}\label{e:qgt}
	G_{\mu\nu}(\theta) := G_{\mu\nu}[\psi_\theta] := \left.\left\langle \frac{\partial \psi_\theta}{\partial \theta^{\mu}}  \right| \frac{\partial \psi_\theta}{\partial \theta^{\nu}} \right\rangle - \left.\left\langle \frac{\partial \psi_\theta}{\partial \theta^{\mu}} \right| \psi_\theta\right\rangle  \left\langle \psi_\theta \left| \frac{\partial \psi_\theta}{\partial \theta^{\nu}}  \right\rangle\right.  \enspace , \quad\quad\quad 1 \leq \mu,\nu \leq n
\end{equation}
whose tensorial property is confirmed by elementary calculus. The quantum geometric tensor is manifestly invariant under unitary transformations,
\begin{equation}
    G[U(\psi_\theta)] = G[\psi_\theta]
\end{equation}
and is additionally invariant under a local phase transformation of the parametrized family which has no classical analogue.
In particular, for any differentiable function $\omega : \mathbb{R}^n \to \mathbb{R}$,
\begin{equation}\label{e:phaseinv}
    G[\exp(i\omega(\theta))\psi_\theta] = G[\psi_\theta] \enspace .
\end{equation}
The real part of the quantum geometric tensor, $g(\theta) := \operatorname{Re}[G(\theta)]$ is necessarily a real symmetric positive semi-definite matrix\footnote{Which, incidentally, equals four times the quantum Fisher information matrix for pure states.} $g(\theta) \in \mathbb{S}^n_+$.
In direct analogy to the information matrix \eqref{e:FIM}, the matrix $g(\theta)$ is the coefficient matrix for the infinitesimal line element obtained by restricting the Fubini-Study distance to the parametric family
\begin{equation}
    d_{\rm FS}^2(P_{\psi_\theta},P_{\psi_{\theta + {\rm d}\theta}}) = \sum_{\mu,\nu=1}^n g_{\mu\nu}(\theta) {\rm d}\theta^{\mu} {\rm d}\theta^{\nu} \enspace .
\end{equation}
By an abuse of terminology, we will refer to $g(\theta)$ as a metric tensor.
In the special case when $\psi_\theta(x)$ is real-valued, classical and quantum information geometry coincide in the sense that
\begin{equation}\label{e:classicalquantum}
    g[\psi_\theta] = \frac{1}{4} I[\psi_\theta^2] \enspace .
\end{equation}

\subsection{Unnormalized wavefunctions and holomorphy}
It is often convenient to represent the unit-normalized family of functions $\psi_\theta$ via another family of functions $\Psi_{\theta} \in L^2(\mathbb{R}^d)$ whose normalization is unknown,
\begin{equation}\label{e:unnormalized}
    \psi_\theta(x) = \frac{\Psi_{\theta}(x)}{\sqrt{\langle \Psi_{\theta} | \Psi_{\theta} \rangle}} \enspace .
\end{equation}
If the quantum geometric tensor is expressed in terms of $\Psi_{\theta}$ one obtains the following useful expression,
\begin{equation}\label{e:qgt_unnormalized}
	G_{\mu\nu}(\theta) = \frac{1}{\langle \Psi_{\theta} | \Psi_{\theta} \rangle}\left.\left\langle \frac{\partial \Psi_{\theta}}{\partial \theta^{\mu}}  \right| \frac{\partial \Psi_{\theta}}{\partial \theta^{\nu}} \right\rangle -
	\frac{1}{\langle \Psi_{\theta} | \Psi_{\theta} \rangle^2}\left.\left\langle \frac{\partial \Psi_{\theta}}{\partial \theta^{\mu}} \right| \Psi_{\theta}\right\rangle  \left\langle \Psi_{\theta} \left| \frac{\partial \Psi_{\theta}}{\partial \theta^{\nu}}  \right\rangle\right.
	\enspace , \quad\quad\quad 1 \leq \mu,\nu \leq n
\end{equation}
In many important applications, the unnormalized function $\Psi_{\theta}$ is parametrized in terms of an even number $n = 2m$ of real parameters $\theta = \theta_1 \oplus \theta_2$ satisfying the following differential identities which correspond to the Cauchy-Riemann equations for the components of the complex vector $z :=\theta_1 + i \theta_2 \in \mathbb{C}^m$,
\begin{equation}\label{e:holo}
\nabla_{\theta_2} \Psi_\theta = i \nabla_{\theta_1}\Psi_\theta
\end{equation}
which is equivalent to the vanishing of the Wirtinger gradient of $\Psi_\theta$ with respect to $\overline{z}$,
\begin{equation}
    \nabla_{\overline{z}} \Psi_\theta := \frac{1}{2}\left(\nabla_{\theta_1}+i\nabla_{\theta_2}\right)\Psi_\theta = 0 \enspace .
\end{equation}
In this case the quantum geometric tensor and its real part are respectively given by,
\begin{equation}\label{e:holoG}
    G(\theta) =
\begin{bmatrix}
    S(z) & i S(z) \\
    -iS(z) & S(z)
\end{bmatrix}
\enspace , \quad \quad 
g(\theta) =
\begin{bmatrix}
    \operatorname{Re}[S(z)] & -\operatorname{Im}[S(z)] \\
    \operatorname{Im}[S(z)] & \operatorname{Re}[S(z)]
\end{bmatrix}
\enspace ,
\end{equation}
where $S(z)\in \mathbb{H}^m_+$ is the Hermitian postive semi-definite sub-block of the quantum geometric tensor corresponding to the $\theta_1$ axis, whose components can be expressed in terms of Wirtinger derivatives as follows,
\begin{equation}
    S_{ij}(z) = \frac{1}{\langle \Psi_{\theta} | \Psi_{\theta} \rangle}\left.\left\langle \frac{\partial \Psi_{\theta}}{\partial z^i}  \right| \frac{\partial \Psi_{\theta}}{\partial z^j} \right\rangle - 
	\frac{1}{\langle \Psi_{\theta} | \Psi_{\theta} \rangle^2}\left.\left\langle \frac{\partial \Psi_{\theta}}{\partial z^i} \right| \Psi_{\theta}\right\rangle  \left\langle \Psi_{\theta} \left| \frac{\partial \Psi_{\theta}}{\partial z^j}  \right\rangle\right. \enspace  , \quad \quad \quad 1\leq i,j \leq m
\end{equation}
\section{Time-dependent variational principles}\label{sec:TDVPs}
With the above information-geometric preliminaries we are prepared to discuss variational principles. Consider the evolution of the unit-normalized function $\psi_\theta$ under the operator $\exp(-A \delta t)$ where $A$ denotes either the Hermitian Hamiltonian $H$ or the skew-Hermitian operator $iH$ and $\delta t \geq 0$ denotes the evolution time. The unconstrained evolution $\psi_\theta \mapsto \exp(-A\delta t) \psi_\theta$ generically produces a function outside of the set $\{\psi_\theta\}_{\theta \in \mathbb{R}^n}$.
Now consider the constrained evolution $\psi_{\theta} \mapsto \psi_{\theta + \delta\theta}$ induced by a parameter shift $\delta\theta \in \mathbb{R}^n$. The optimal shift should be chosen to minimize the Fubini-Study distance between the associated projectors \cite{carleo2017solving}. Thus, given an initial parameter vector $\theta_0 \in \mathbb{R}^n$ and a step size $\delta t > 0$, define a sequence of parameter vectors $( \theta_k )_{k \in \mathbb{N} }$ by the following iteration
\begin{align}\label{e:projected}
    \theta_{k+1}
    & = \theta_k + \delta \theta_k \\
    \delta \theta_k
    & := \argmin_{\delta\theta \in \mathbb{R}^n} d_{\rm FS}\big(P_{\exp(-A\delta t)\psi_{\theta_k}},P_{\psi_{\theta_k+\delta\theta}}\big)
\end{align}
In the limit $\delta t \to 0$, the sequence approximates the solution of a system of ordinary differential equations with initial condition $\theta(0) = \theta_0$.
In particular, for $A=H$ and $A=iH$ we obtain respectively \cite{stokes2020quantum, barison2021efficient},
\begin{equation}\label{e:tdvp}
g(\theta(t)) \dot{\theta}(t)
=
\begin{cases}
-\operatorname{Re}[F(\theta(t))] \\
\operatorname{Im}[F(\theta(t))]
\end{cases}
\end{equation}
where we have defined the following complex vector $F(\theta) \in \mathbb{C}^n$ for all $\theta \in \mathbb{R}^n$,
\begin{equation}\label{e:force}
     F(\theta) := \left\langle \nabla_\theta \psi_\theta \middle| H \psi_\theta \right\rangle - \left\langle \nabla_\theta \psi_\theta \middle| \psi_\theta \right\rangle \left\langle \psi_\theta \middle| H \psi_\theta \right\rangle
\end{equation}
The arguments leading to the above evolution equations are reviewed in appendix \ref{app:tdvp}. The same equations are obtained by applying McLachlan's variational principle to the Liouville-von Neumann equation restricted to pure states \cite{yuan2019theory}. McLachlan's variational principle applied to the time-dependent Schr\"{o}dinger equation, however, yields a different set of evolution equations \cite{mclachlan1964variational, yuan2019theory}\footnote{See \cite[Appendix A.2.2 and B.2.2]{yuan2019theory} and \cite[Appendix A.1.2 and B.1.2]{yuan2019theory} for the variational forms of the von Neumann and time-dependent Schr\"{o}dinger equation, respectively. Further clarification about these variational principles is presented in appendix \ref{app:mclachlan}.}.
Using the fact that $\langle \psi_\theta | H \psi_\theta \rangle \in \mathbb{R}$ and $\langle \nabla_\theta \psi_\theta | \psi_\theta \rangle \in i\mathbb{
R}$, it follows from \eqref{e:tdvp} that the imaginary-time evolution equation coincides with Riemannian gradient flow in the geometry induced by the Fubini-Study metric \cite{stokes2020quantum},
\begin{equation}\label{e:gradientflow}
    g(\theta(t))\dot{\theta}(t) = -\nabla \mathcal{L}(\theta(t)) \enspace , \quad\quad \mathcal{L}(\theta) := \frac{1}{2}\langle \psi_\theta | H\psi_\theta \rangle \enspace ,
\end{equation}
which implies, as a trivial consequence of the positive semi-definiteness condition $g(\theta) \in \mathbb{S}^n_+$, that the energy $\mathcal{L}(\theta)$ is non-increasing under imaginary-time evolution, 
\begin{equation}
    \dot{\mathcal{L}}(t) = -\dot{\theta}(t)^T g(\theta(t)) \, \dot{\theta}(t) \leq 0 \enspace .
\end{equation}
Furthermore, in the special case of real-valued $\psi_\theta$, \eqref{e:classicalquantum} implies that imaginary-time evolution is a special case of natural gradient flow \cite{amari1998natural}. Finally, since $g(\theta)$ is typically a degenerate metric as discussed in section \ref{sec:informationgeometry} regularization is typically required in order to obtain a well-posed system of ordinary differential equations.
\subsection{Unnormalized wavefunctions and holomorphy}
In the case of unnormalized wavefunctions, the expression for $F(\theta)$ becomes,
\begin{equation}\label{e:force_unnormalized}
F(\theta) = \frac{\langle \nabla_\theta \Psi_{\theta} | H \Psi_{\theta} \rangle}{\langle \Psi_{\theta} | \Psi_{\theta} \rangle} - \frac{\langle \nabla_\theta \Psi_{\theta} | \Psi_{\theta} \rangle}{\langle \Psi_{\theta} | \Psi_{\theta} \rangle}\frac{\langle \Psi_{\theta} | H \Psi_{\theta}\rangle}{\langle \Psi_{\theta} | \Psi_{\theta} \rangle} \enspace .
\end{equation}
Suppose, in addition, that $n = 2m$ is even and that the holomorphic constraints \eqref{e:holo} are satisfied. Let us overload the notation by denoting $F(z) \in \mathbb{C}^m$ the subvector of $F(\theta)\in \mathbb{C}^{n}$ along the $\theta_1$ axis. Then it follows from \eqref{e:holo} and \eqref{e:force_unnormalized} that,
\begin{align}\label{e:holoforce}
F(\theta) & =
\begin{bmatrix}
    F(z) \\
    -iF(z)
\end{bmatrix}
\end{align}
It can then be shown that, under the holomorphic assumption, the evolution equations \eqref{e:tdvp} reduce to the following differential equations for the complex vector $z \in \mathbb{C}^m$,
\begin{equation}\label{e:holotdvp}
S(z(t))\dot{z}(t) =
\begin{cases}
- F(z(t)) \\
-iF(z(t))
\end{cases}
\end{equation}
The instantaneous rate of change of the loss function $\mathcal{L}$ under the evolution equations \eqref{e:holotdvp} is given by 
\begin{equation}\label{e:holorate}
    \dot{\mathcal{L}}(t) =
    \begin{cases}
    - \dot{z}(t)^\dag S(z(t)) \dot{z}(t) \\
    0
    \end{cases}
\end{equation}
Recalling the Hermitian positive semi-definiteness of $S(z) \in \mathbb{H}^m_+$, the above equations imply that the energy is non-increasing or conserved under the imaginary- or real-time evolution respectively.
\section{Stochastic estimation}

\subsection{VMC and t-VMC}
Now we discuss numerical solution of \eqref{e:tdvp} using stochastic estimation, assuming an efficient algorithm to compute $x \mapsto [\psi_\theta(x), \nabla_\theta \psi_\theta(x), (H\psi_\theta)(x)]$ and an efficient algorithm to generate unbiased samples according to the probability density $|\psi_\theta(x)|^2$. Most literature on VMC and time-dependent VMC \cite{carleo2012localization, carleo2014light, carleo2017solving, ido2015time} has focused on the assumption of an efficient mapping $x \mapsto [\Psi_{\theta}(x), \nabla_\theta \Psi_{\theta}(x), (H\Psi_{\theta})(x)]$ and approximate sampling from the density $|\psi_\theta(x)|^2$ using Markov Chain Monte Carlo\footnote{See however \cite{rindaldi2021,han2020deep,xie2021ab} for continuum and \cite{sharir2020deep,hibat2020recurrent} for discrete quantum systems, respectively.}. Define the Born probability density $\rho_\theta(x) \in [0,\infty)$, the wavefunction score $\sigma_\theta(x) \in \mathbb{C}^n$ and the local energy $l_\theta(x) \in \mathbb{C}$ as follows,
\begin{equation}\label{e:local_energy}
    \rho_\theta(x) := |\psi_\theta(x)|^2 \enspace , \quad \quad \sigma_\theta(x) := \frac{\nabla_\theta \psi_\theta(x)}{\psi_\theta(x)} \enspace , \quad\quad l_\theta(x) := \frac{(H\psi_\theta)(x)}{\psi_\theta(x)} \enspace .
\end{equation}
It is then a simple exercise to confirm that the quantities $g(\theta)$ and $F(\theta)$, $\mathcal{L}(\theta)$ and $\nabla \mathcal{L}(\theta)$ can be expressed as the expectation values of random variables with respect to the Born probability density. In particular,
\begin{align}\label{e:gF}
    G(\theta) 
    = \cov(\sigma_\theta,\sigma_\theta)^T \enspace , \quad\quad\quad
    F(\theta)
    = \cov(l_\theta,\sigma_\theta)^T \enspace ,
\end{align}
The expressions \eqref{e:gF} provide the basis for a stochastic approximate solution of \eqref{e:tdvp} called stochastic reconfiguration and t-VMC, respectively, which use Monte Carlo methods to approximate the covariances, in combination with a suitable time-marching scheme (e.g., forward Euler). In the special case of imaginary-time evolution, the algorithm is known as stochastic reconfiguration \cite{sorella2007weak}. If, in addition $\psi_\theta \in \mathbb{R}$ for all $\theta \in \mathbb{R}^n$ then, stochastic reconfiguration becomes stochastic natural gradient descent for the objective $\mathcal{L}(\theta)$.

\subsection{Stochastic optimization and variance reduction}\label{sec: stochastic_optimization_and_variance_reduction}
Stochastic reconfiguration can be understood as a special case of stochastic gradient-based optimization for the objective $\mathcal{L}(\theta)$. To see this, let 
$\hat{\mathcal{L}}_\theta(x)$ be any unbiased estimator for the loss function so that
\begin{align}\label{e:loss}
    \mathcal{L}(\theta) = \mathbb{E}\left[ \hat{\mathcal{L}}_\theta(x) \right] \enspace .
\end{align}
Then by the log-derivative trick,
\begin{equation}\label{e:stochoptim}
    \nabla \mathcal{L}(\theta) = \mathbb{E}\left[
    \left(\hat{\mathcal{L}}_\theta(x)\mathbbm{1} - \frac{B}{2}\right)\nabla_\theta \log \rho_\theta(x) \right]
    +
    \mathbb{E}\left[\nabla_\theta
    \hat{\mathcal{L}}_\theta(x)\right] \enspace ,
\end{equation}
where $B \in \mathbb{R}^{n \times n}$ is an arbitrary matrix and we have used the fact that $\mathbb{E} [\nabla_\theta\log \rho_\theta(x)] = 0$.  The canonical estimator which has been pursued most widely in the literature corresponds to the local energy defined in \eqref{e:local_energy},
\begin{equation}\label{e:canonical}
    \hat{\mathcal{L}}_\theta^{\rm (can)}(x) := \frac{1}{2}l_\theta(x) \enspace .
\end{equation}
The canonical estimator has a number of desirable properties including the fact the gradient of the objective \eqref{e:stochoptim} becomes independent of the gradient of the estimator. Specifically, plugging \eqref{e:canonical} into \eqref{e:stochoptim} and using Hermiticity of $H$ we obtain,
\begin{equation}\label{e:canongrad}
    \nabla \mathcal{L}(\theta) = \operatorname{Re} \mathbb{E} \big[ \big(l_\theta(x)\mathbbm{1} - B \big) \overline{\sigma_\theta}(x) \big] \enspace ,
\end{equation}
which coincides with stochastic reconfiguration for the choice of baseline $B = \mathbb{E}[l_\theta(x)]$.
In addition, the variance of the stochastic objective  function using the canonical estimator is proportional to the quantum variance of the Hamiltonian in the quantum state $P_{\psi_\theta}$,
\begin{equation}
    \var(l_\theta) := \cov(l_\theta,l_\theta) = \langle \psi_\theta | H \psi_\theta \rangle - \langle \psi_\theta | H^2 \psi_\theta \rangle \enspace ,
\end{equation}
which follows from Hermiticity of $H$.
The canonical estimator thus has the desirable property that its variance approaches zero when $\psi_\theta$ approaches any eigenfunction of $H$. This zero-variance principle is an attractive feature of stochastic reconfiguration compared to other stochastic optimization problems, and can be exploited when numerically approaching a ground eigenfunction. It turns out that the canonical estimator is not the only estimator available in continuous-variable VMC, however, and we will discuss one such alternative in the experiments section.

For normalized trial functions, it has been shown empirically \cite{sharir2020deep, han2020deep, hibat2020recurrent, xie2021ab} that the stochastic reconfiguration choice of baseline $B=\mathbb{E}[l_\theta(x)]$ has reduced variance compared to vanishing baseline. In order to better understand this variance reduction property, let us define the following gradient estimator with arbitrary baseline $B$,
\begin{equation}
    \hat{\nabla}_{\theta,B}(x)
    := \operatorname{Re}\big[\big( l_\theta(x)\mathbbm{1} - B\big) \overline{\sigma_\theta}(x)\big] \enspace .
\end{equation}
Anticipating the variance-reduction property of the baseline, introduce a convex objective function for the matrix $B$; namely, the total variation of the random vector $\hat{\nabla}_{\theta,B}$,
\begin{equation}\label{e:Bloss}
    V(B) := \tr \big[\cov(\hat{\nabla}_{\theta,B})\big] \enspace .
\end{equation}
The stationary points of \eqref{e:Bloss} describe a linear system of equations for the matrix $B$. In appendix \ref{app:baseline}, it is shown that the system can be solved under the ad-hoc assumption that $l_\theta$ and $\sigma_\theta$ are statistically independent under $\rho_\theta$. In particular, one finds that $B$ is an approximate multiple of the identity matrix, $B \approx \mathbb{E}[l_\theta(x)] \mathbbm{1}$. This provides a heuristic justification for the variance reduction procedure utilized in \cite{han2020deep,hibat2020recurrent}.

\section{Additional comments on quantum flows}\label{sec:comments}
 \subsection{Group equivariant dynamics from flows}
An attractive feature of normalizing flows is that they are compatible with group symmetries in the following sense.
\begin{lem}\label{lem:rep}
Let $G \leq O(d,\mathbb{R})$ be an orthogonal matrix group and $\rho : G \to \mathbb{C}^\times$ a one-dimensional representation of $G$. Suppose that $0 \neq \psi \in L^p(\mathbb{R}^d)$ and $f \in \Diff(\mathbb{R}^d)$  transform as
\begin{align}
    \psi(gx)  = \rho(g) \, \psi(x) \enspace , \quad\quad\quad f(gx) = g \, f(x) \enspace , & \quad\quad\quad \textrm{for all $g\in G$ and $x \in \mathbb{R}^d$} \\
\intertext{Then $\rho$ is a unitary representation and $f\cdot\psi \in L^p(\mathbb{R}^d)$ transforms as}
    (f\cdot\psi)(gx) = \rho(g) \, (f\cdot\psi)(x) \enspace , & \quad\quad\quad \textrm{for all $g\in G$ and $x \in \mathbb{R}^d$} \enspace .
\end{align}
\end{lem}
The proof is an elementary generalization of \cite[Lemma 1]{papamakarios2021normalizing}. The above lemma has important implications for quantum simulation because it enables to model a flexible family of normalized $G$-equivariant functions using a simple normalized $G$-equivariant base function $\psi \in L^2(\mathbb{R}^d)$ and a family of $G$-equivariant diffeomorphisms. The set of such $G$-equivariant functions forms a Hilbert subspace of $L^2(\mathbb{R}^d)$, which is stable under real- or imaginary time evolution by a $G$-equivariant Hamiltonian $H$ satisfying,
\begin{equation}
    U_g \circ H = H \circ U_g \enspace, \quad\quad\quad \textrm{for all $g\in G$} \enspace ,
\end{equation}
where the unitary operator $U_g$ acts on the Hilbert space as $U_g(\psi) = \psi(g^{-1}x)$ for all $\psi \in L^2(\mathbb{R}^d)$ (recalling that $|{\det g}|=1$).
If combined with the time-dependent variational principle, the above construction enables the approximation of real- or imaginary-time dynamics within the $G$-equivariant subspace of states. In other words, flows can describe dynamically closed superselection sectors of the Hilbert space. This possibility was recently investigated in \cite{xie2021ab} using the sign representation of a permutation subgroup $G\leq O(d,\mathbb{R)}$ for the purpose of modeling spinless fermions.

\subsection{Universal approximation}
Although universal approximation has been proven under the assumption of strict positivity \cite{papamakarios2021normalizing}, complex-valued flows pose additional subtleties. In particular, as already noted in \cite{xie2021ab},  normalizing flows cannot change the topology of the zero level set for the base function and similarly for the level sets of complex phase. More precisely, the level sets are diffeomorphic,
\begin{equation}\label{e:levelsets}
    L_0(\operatorname{mod}f\cdot\psi) = f \big( L_0(\operatorname{mod}\psi)\big) \enspace, \quad \quad \quad \forall\theta \in [0,2\pi) : L_\theta(\arg f \cdot\psi) = f\big( L_\theta(\arg \psi) \big) \enspace ,
\end{equation}
where $\operatorname{mod}$ denotes the complex modulus and $L_c(\cdot)$ denotes the level set of a function corresponding to the real value $c \in \mathbb{R}$.
 Although the above identities pose an obstruction to universal approximation, this limitation can be easily overcome in practice by promoting the base to a trainable function, or by multiplying the output of the flow with a learnable complex phase.

\section{Experiments} \label{sec:experiments}
\subsection{States and Hamiltonian}
 Let $p = (1/i) \nabla $ denote the momentum operator canonically conjugate to $x$ and
consider the following Hamiltonian operator represented on a suitable subspace of the  Hilbert space $L^2(\mathbb{R}^d)$ by
\begin{align}\label{e:ham}
    H 
        & = \frac{1}{2} 
     \begin{bmatrix}
     x \\ p
     \end{bmatrix}\sbullet
     h
     \begin{bmatrix}
     x \\ p
     \end{bmatrix}
    +
    \frac{1}{4!}\sum_{ijkl} \lambda_{ijkl} \,x_i x_j x_k x_l
    \enspace ,
\end{align}
where $h \in\mathbb{S}^{2d}$ is a symmetric matrix, $\lambda_{ijkl}$ are the components of a real tensor $\lambda \in \mathbb{R}^{d \times d \times d \times d}$ and $\sbullet$ denotes the standard inner product for Euclidean space $\mathbb{R}^{2d}$. The physical significance of the above Hamiltonian is that it describes an indefinite number of bosons occupying $d$ possible modes\footnote{A motivating example from field theory is provided in appendix \ref{app:field}.}.
Additional assumptions about $h$ and $\lambda$ are required to ensure the existence of a ground eigenfunction. If, for example $\lambda=0$, then we require $h \in \mathbb{S}^{2d}_{++}
$ to be symmetric positive definite. In this case, the exact ground eigenfunction is represented by a family of trial functions $\psi_{\rm G} \in L^2(\mathbb{R}^d)$ of the following Gaussian form\footnote{
See 
appendix \ref{app:free} for an example in one dimension.},
\begin{equation}\label{eq: gaussian_wave_fcn}
    \psi_{\rm G}(x) = \left(\det\frac{A}{\pi}\right)^{1/4}  \exp\left[-\frac{1}{2}(x-\mu)^T (A + iB) (x-\mu)\right]
    \enspace ,
\end{equation}
where the variational parameters are constrained to the manifold $(\mu,A,B)\in \mathbb{R}^d \times \mathbb{S}_{++}^{d} \times \mathbb{S}^d$.
Consider the following block decomposition of the symmetric matrix $h \in \mathbb{S}^{2d}$
\begin{equation}
    h
    =
    \begin{bmatrix}
    h_{xx} & h_{xp} \\
    h_{px} & h_{pp}
    \end{bmatrix}
    \enspace ,
\end{equation}
where $h_{px} = h_{xp}^T$ and where $h_{xx},h_{pp} \in \mathbb{S}^d$ are symmetric.

The experiments focus on the problem of approximating a ground eigenfunction via the method of natural gradient descent, using a non-canonical estimator of the gradient obtained using an adjoint representation of the Hamiltonian, the details of which are deferred to appendix \ref{sec:ibp_appendix}. For simplicity, we exemplify the method in the special case of $h_{pp}=\mathbbm{1}$ and $h_{xp} = 0$, where 
it should be stressed that
this example is only for illustrative purposes since quantum Monte
 Carlo does not suffer from a sign-problem when $h_{xp}=0$. In addition, since we are targeting the ground eigenfunction rather than the full imaginary-time trajectory, we may employ any step size schedule consistent with descreasing energy. In particular, we combined the natural gradient method with the Adam optimizer.
In order to ensure boundedness of the Hamiltonian from below, we chose the interaction tensor $\lambda_{ijkl}=3\delta_{ij}\delta_{kl} u_{ik}$ for some symmetric positive definite matrix $u \in \mathbb{S}^d_{++}$.  The resulting Hamiltonian has a unique ground space spanned by a strictly positive eigenfunction \cite[Section 3.3]{glimm2012quantum}. By positivity of the ground function, if the Hamiltonian is $G$-equivariant for some $G \leq O(d,\mathbb{R})$, then it follows that the ground eigenfunction transforms in the trivial representation. Thus, provided that
$G$ is a finite group of computationally tractable order, a trial function for the $G$-invariant ground eigenfunction can be chosen as the square root of a mixture density formed by symmetrizing a classical normalizing flow $p_{\theta}$ as follows,
\begin{equation}
    \psi_\theta(x)^2 = \frac{1}{|G|} \sum_{g\in G} p_\theta(gx) \enspace . \label{eq: symmetrized_flow}
\end{equation}
In the experiments the matrices $u$ and $h_{xx}$ were chosen randomly to reflect our agnosticism about the nature of the regularization. It follows that the only remaining symmetry of the Hamiltonian is $G=\mathbb{Z}_2$, whose non-trivial element is implemented by field space inversion $x \mapsto -x$. Since $\mathbb{Z}_2$ is a finite group of order 2, the mixture density approach is computationally efficient in this case.

\subsection{Experimental setup}
Recall that under the above simplifications the Hamiltonian is specified by a symmetric matrix $h_{xx} \in \mathbb{S}^{d}$ and a symmetric positive-definite matrix $u \in \mathbb{S}^d_{++}$.
For each problem dimension $d \in \{2,5,10,25,50,100\}$ a random Hamiltonian was selected by choosing the $h$ and $\lambda$ parameters in \eqref{e:ham} subject to the constraints described in the previous subsection, specifically that $h$ is a symmetric block diagonal matrix of the form $h=\diag(h_{xx},\mathbbm{1})$ with $h_{xx}$ symmetric negative definite\footnote{That is, $-h_{xx} \in \mathbb{S}^d_{++}$.} and that $\lambda_{ijkl}=3\delta_{ij}\delta_{kl} u_{ik}$ for symmetric positive definite $u \in \mathbb{S}^d_{++}$. The relevant matrix ensembles and sampling procedures are described in  Appendix \ref{sec:generating_u_and_h}. 

For each randomly selected Hamiltonian, the ground eigenfunction is approximated by representing it using a $\mathbb{Z}_2$-symmetrical trial function of the form \eqref{eq: symmetrized_flow} where $p_\theta$ was chosen to be a RealNVP-based normalizing flow \cite{realnvp2017} with variational parameters $\theta \in \mathbb{R}^n$. Starting from random initialization, the parameters $\theta$ of the neural network were updated using stochastic natural gradient-based optimization of the loss function \eqref{e:loss}, employing the adjoint representation of the Hamiltonian. Stochastic estimates of the gradient on each iteration were obtained using PyTorch to back-propagate \eqref{e:loss}, without the use of baseline adjustment. Additional details on the normalizing flow architecture and training procedure can be found in Appendix \ref{sec: architecture}. Following the recent work of \cite{hackett2021flow}, adiabatic retraining was also explored (Appendix \ref{sec: adiabatic_retraining} for additional detail). 

As a baseline, the results are compared to a real-valued Gaussian trial state \eqref{eq: gaussian_wave_fcn}, for which the variational energy can be computed analytically (see Appendix \ref{sec: gaussian_state_approx}). Optimization over the parameter manifold $(\mu, A) \in \mathbb{R}^d \times \mathbb{S}^d_{++}$ was performed using the Pymanopt toolbox \cite{pymanopt}.

\subsection{Results}

The energies found via the normalizing flow approaches and the Gaussian approximation are shown in Figure \ref{fig:boxplot}. We can infer that the normalizing flow methods generally lead to lower estimates of the ground state energy than the Gaussian state approximation. In addition, compared to unsymmetrized flows, symmetrized flows with adiabatic retraining yield lower energies in higher dimensions. In Figure \ref{fig:histograms_combined}, the potential function and the approximations of the ground-state probability density found via normalizing flows or the Gaussian wave approximation are visualized for problem dimension $d=2$. Notice that while the symmetrized normalizing flow is able to find the two expected modes for the given the potential function, the unsymmetrized flow and Gaussian approximation collapse around one of the modes.
\begin{figure}
    \centering
    \includegraphics[width=\textwidth]{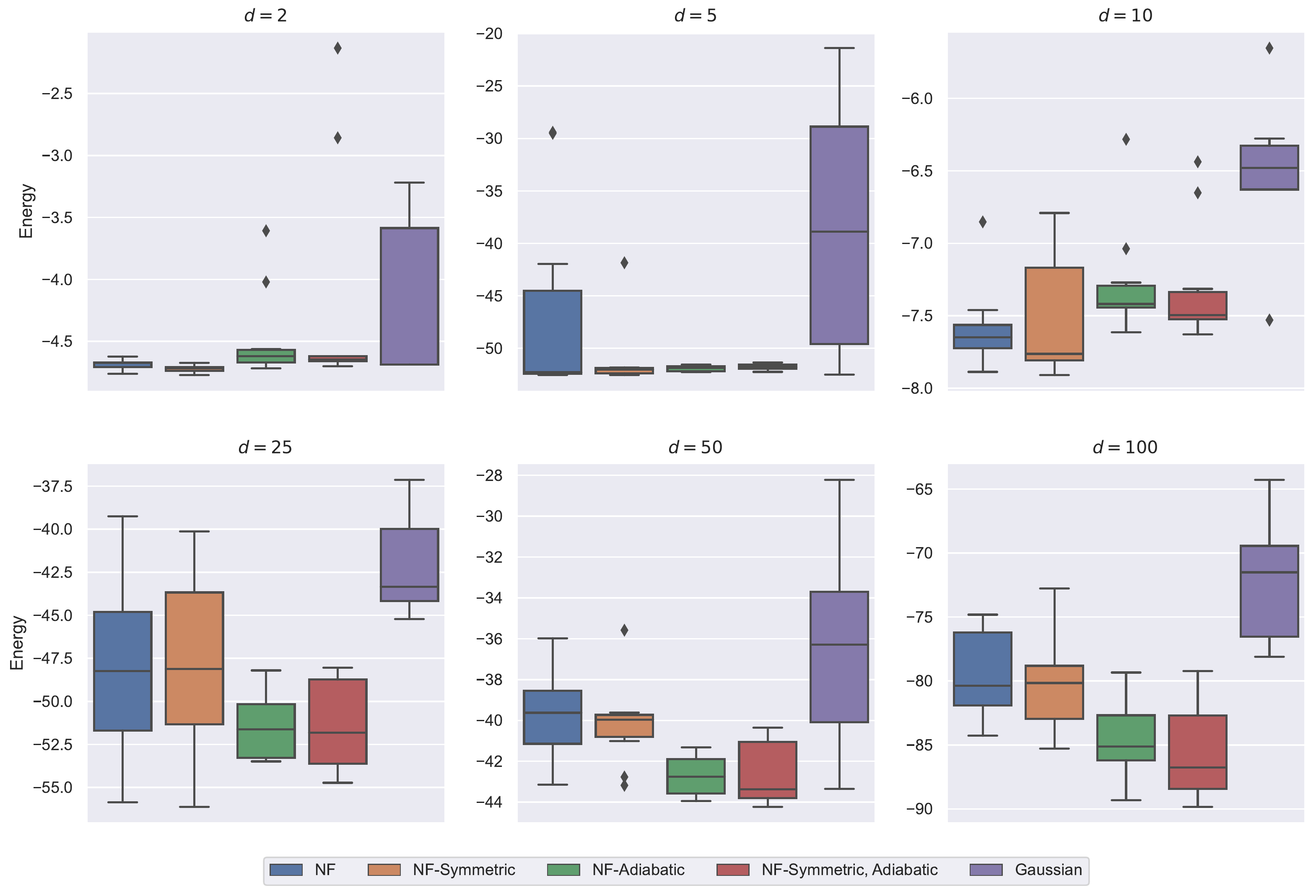}
    \caption{\textit{A boxplot showing the energies found via normalizing flows and Gaussian state approximation (results shown for ten random initializations). ``NF" and ``NF-Symmetric" refer to the unsymmetrized and symmetrized normalizing flows, respectively (see \eqref{eq: symmetrized_flow}). The results labeled with ``adiabatic" are trained using adiabatic retraining, as described in Appendix \ref{sec: adiabatic_retraining}.}}
    \label{fig:boxplot}
\end{figure}

\begin{figure}
    \centering
    \includegraphics[width=\textwidth]{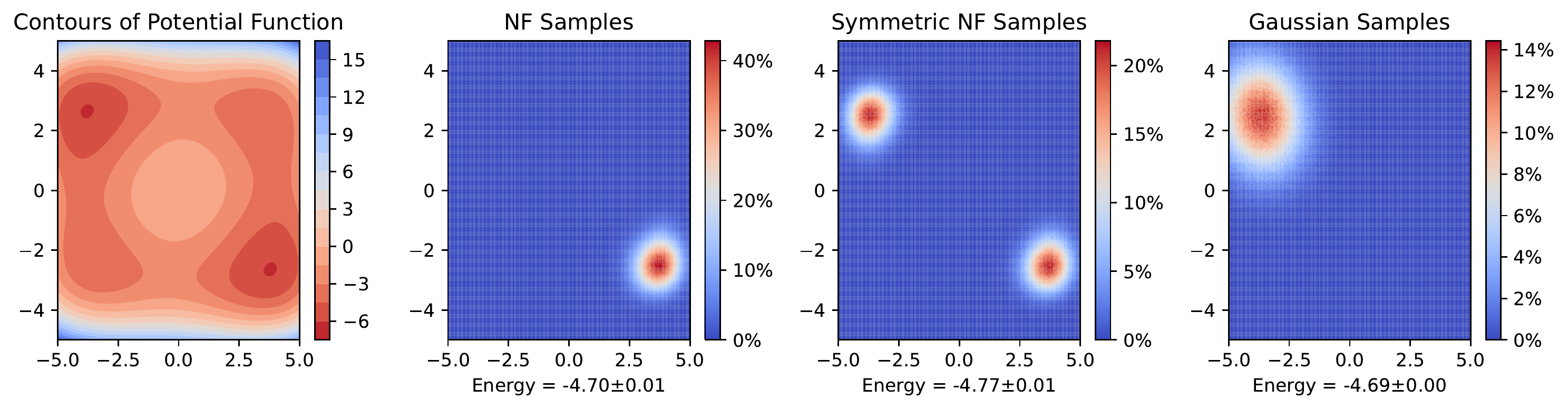}
    \caption{\textit{From left to right: Contours of the potential function and density plots from the corresponding optimized wave functions for $d=2$ (unsymmetrized normalizing flows, symmetric normalizing flows, and Gaussian wave function approximation).}}
    \label{fig:histograms_combined}
\end{figure}

\section{Discussion and future directions}
\label{sec:conclusions}
In summary, flow-based parametrizations provide a promising class of trial wavefunctions for use in the continuous-variable VMC, although a number of challenges remain to be solved. Let us conclude by summarizing open problems. An obvious limitation of flow-based parametrizations is their lack of holomorphy, which is required to ensure energy conservation during variational real-time evolution. It would be very interesting to reconcile the holomorphic constraint with the exact sampling property of flows in the continuum. Since VMC methods suffer from finite sampling effects, it will be interesting to undertake finite-sample analysis of the different gradient estimators and to further explore variance reduction strategies. It would also be interesting to better understand the approximation power of $G$-equivariant flows. 

It will be natural to extend the experiments to address the sign-problem, including the determination of ground states for non-stoquastic systems characterized by $h_{xp} \neq 0$, or real-time evolution.  In particular, it would also be interesting to numerically investigate the extent to which energy conservation is violated in practice.

An important outstanding challenge is to find physical applications of the scheme which exhibit an advantage compared to the best known algorithm. Aspirational targets include charged scalar fields at finite chemical potential \cite{weir2010studying} or variational calculation of scattering \cite{liu2021towards}.

\section{Acknowledgements}

J.S.~thanks Giuseppe Carleo, Damian Hofmann, Di Luo, Matija Medvidovi\'{c}, Jannes Nys, and Gabriel Pescia for helpful discussions. Authors gratefully acknowledge support from NSF under grant DMS-2038030. This research was supported in part through computational resources and services provided by the Advanced Research Computing (ARC) at the University of Michigan.

\newpage
\begin{appendix}
\section{McLachlan's variational principles}\label{app:mclachlan}
This section constructs an (admittedly contrived) example demonstrating the failure of McLachlan's variational principle for the time-dependent Sch\"{o}dinger equation. This example complements \cite[Section 4.2]{yuan2019theory} which considered real time evolution of a single qubit. Consider the quantum simple harmonic oscillator
\begin{equation}\label{eq:simple_harmonic_oscillator}
    H = \frac{1}{2}\big(p^2 + x^2\big)
\end{equation}
initialized in the state given by projection onto the ground eigenfunction,
\begin{equation}\label{eq:simple_harmonic_oscillator_ground_eigenfcn}
    \psi_0(x) = \frac{1}{\pi^{1/4}} \exp\left[-\frac{1}{2} x^2\right]
\end{equation}
The exact projector dynamics under real-time evolution is clearly given by the trivial constant sequence $(P_{\psi_0})_{t \geq 0}$. Now consider the following normalized variational family
\begin{equation}\label{e:d=1}
    \psi_{\theta}(x) = \left(\frac{a}{\pi}\right)^{1/4} \exp\left[-\frac{1}{2}(a + ib) x^2\right]
\end{equation}
parametrized in terms of $\theta := (\log a, b) \in \mathbb{R}^2$, which ensures the positivity constraint $a>0$. Recall that variational Liouville-von Neumann equation and the variational time-dependent Schr\"{o}digner equation (TDSE) are given, respectively by
\begin{equation}
    g(\theta(t)) \dot{\theta}(t) = \operatorname{Im}[F(\theta(t))] \enspace , \quad \quad \quad \widetilde g(\theta(t)) \dot{\theta}(t) = \operatorname{Im}[\widetilde F(\theta(t))]
\end{equation}
where
\begin{align}
    F_\mu(\theta)
    & = \left\langle \frac{\partial \psi_\theta}{\partial \theta^\mu} \middle| H \psi_\theta \right\rangle - \left\langle \frac{\partial \psi_\theta}{\partial \theta^\mu} \middle| \psi_\theta \right\rangle \langle \psi_\theta | H \psi_\theta \rangle & \widetilde F_\mu(\theta)
    & = \left\langle \frac{\partial \psi_\theta}{\partial \theta^\mu} \middle| H \psi_\theta \right\rangle \\
    g_{\mu\nu}(\theta) 
    & = \operatorname{Re}\left[\left.\left\langle \frac{\partial \psi_\theta}{\partial \theta^{\mu}}  \right| \frac{\partial \psi_\theta}{\partial \theta^{\nu}} \right\rangle - \left.\left\langle \frac{\partial \psi_\theta}{\partial \theta^{\mu}} \right| \psi_\theta\right\rangle  \left\langle \psi_\theta \left| \frac{\partial \psi_\theta}{\partial \theta^{\nu}}  \right\rangle\right.\right] & \widetilde g_{\mu\nu}(\theta) 
    & = \operatorname{Re}\left[\left.\left\langle \frac{\partial \psi_\theta}{\partial \theta^{\mu}}  \right| \frac{\partial \psi_\theta}{\partial \theta^{\nu}} \right\rangle \right]
\end{align}
In the example above we compute
\begin{align}
    \langle \psi_\theta | H \psi_\theta \rangle 
    & = \frac{1 + a^2 + b^2}{4a^2} \\
    \operatorname{Im}[\widetilde F(\theta)]
    & = 
    \begin{bmatrix}
    \frac{b}{4} \\
    -\frac{1}{16} + \frac{3(1+b^2)}{16a^2}
    \end{bmatrix} \\
    \operatorname{Im}[F(\theta)]
    & = 
    \begin{bmatrix}
    \frac{b}{4} \\
    -\frac{1}{16} + \frac{3(1+b^2)}{16a^2}
    \end{bmatrix}
    -
    \begin{bmatrix}
    0 \\
    \frac{1}{16} + \frac{1+b^2}{16a^2}
    \end{bmatrix}
    \\
    & = \begin{bmatrix}
    \frac{b}{4} \\
    -\frac{1}{8} + \frac{1+b^2}{8a^2}
    \end{bmatrix} \\
    \widetilde g(\theta)
    & =
    \begin{bmatrix}
    \frac{1}{8} & 0 \\
    0 & \frac{3}{16a^2}
    \end{bmatrix} \\
    g(\theta)
    & =
    \begin{bmatrix}
    \frac{1}{8} & 0 \\
    0 & \frac{3}{16a^2}
    \end{bmatrix}
    -
    \begin{bmatrix}
    0 & 0 \\
    0 & \frac{1}{16a^2}
    \end{bmatrix} \\
    & =
    \begin{bmatrix}
    \frac{1}{8} & 0 \\
    0 & \frac{1}{8a^2}
    \end{bmatrix}
\end{align}
Putting the above pieces together we obtain the following variational Liouville-von Neumann equation
\begin{equation}\label{e:vN}
    \frac{{\rm d}}{{\rm d} t}
    \begin{bmatrix}
    \log a \\
    b
    \end{bmatrix}
    =
    \begin{bmatrix}
    2b \\
    1-a^2 + b^2
    \end{bmatrix}
\end{equation}
and the following variational TDSE,
\begin{equation}\label{e:TDSE}
    \frac{{\rm d}}{{\rm d} t}
    \begin{bmatrix}
    \log a \\
    b
    \end{bmatrix}
    =
    \begin{bmatrix}
    2b \\
    1-\tfrac{1}{3}a^2 + b^2
    \end{bmatrix}
\end{equation}
If the system \eqref{e:vN} is initialized at the origin of the $\theta = (\log a, b) \in \mathbb{R}^2$ coordinates, then it will clearly reproduce the exact projector dynamics since the right-hand side of \eqref{e:vN} vanishes. In contrast, it can be verified that there is no choice of initialization which produces the exact projector under \eqref{e:TDSE}.

\section{Justification for choice of baseline}\label{app:baseline}
In this appendix argue that the stationary points of \eqref{e:Bloss} are approximately solved by the stochastic reconfiguration baseline.
Starting from the definition of the loss function \eqref{e:Bloss} and recalling that $\hat\nabla_{\theta,B}(x)$ is real-valued,
\begin{align}
    V(B)
    & = \tr \big[\cov\big(\hat\nabla_{\theta,B}\big)\big] \\
    & = \tr \left\{ \mathbb{E}\left[\hat\nabla_{\theta,B}(x) \otimes \hat\nabla_{\theta,B}(x)\right] - \mathbb{E}\left[\hat\nabla_{\theta,B}(x) \right]\otimes\mathbb{E}\left[ \hat\nabla_{\theta,B}(x)\right] \right\} \\
     & = \tr \left\{ \mathbb{E}\left[\hat\nabla_{\theta,B}(x) \otimes \hat\nabla_{\theta,B}(x)\right] - \nabla \mathcal{L}(\theta) \otimes \nabla\mathcal{L}(\theta) \right\} \\
     & = \mathbb{E}\left[\hat\nabla_{\theta,B}(x)^T \hat\nabla_{\theta,B}(x)\right] - \nabla \mathcal{L}(\theta)^T \nabla\mathcal{L}(\theta)
\end{align}
Now using the fact that the second term above is manifestly independent of $B$ we obtain,
\begin{align}
    \frac{\partial V(B)}{\partial B}
    & =  \frac{\partial }{\partial B}\mathbb{E}\left[ \hat\nabla_{\theta,B}(x)^T \hat\nabla_{\theta,B}(x) \right] \\
    & = -2\mathbb{E}\left[\hat\nabla_{\theta,B}(x) \otimes \operatorname{Re}\big[\overline{\sigma_\theta}(x)\big] \right] \\
    & = -2\mathbb{E}\Big[\operatorname{Re}\left[\big(l_\theta(x) - B\big)\overline{\sigma_\theta}(x)\right] \otimes \operatorname{Re}\big[\overline{\sigma_\theta}(x)\big] \Big] \\
    & = -2\mathbb{E}\Big[\operatorname{Re}\left[\big(\overline{l_\theta}(x) - B\big)\sigma_\theta(x)\right] \otimes \operatorname{Re}\big[\sigma_\theta(x)\big] \Big]
\end{align}
Setting the gradient to zero we obtain,
\begin{equation}\label{e:optimalbaseline}
    B \, \mathbb{E} \Big[ \operatorname{Re}[\sigma_\theta(x)] \otimes \operatorname{Re}[\sigma_\theta(x)] \Big] = \mathbb{E} \Big[ \operatorname{Re}[\overline{l_\theta}(x) \sigma_\theta(x)]  \otimes \operatorname{Re}[\sigma_\theta(x)]\Big]
\end{equation}
Now the right-hand side is
\begin{align}
    \mathbb{E} \Big[ \operatorname{Re}[\overline{l_\theta}(x)\sigma_\theta(x)]  \otimes \operatorname{Re}[\sigma_\theta(x)]\Big]
    & =
    \frac{1}{2}\mathbb{E} \Big[ \left(\overline{l_\theta}(x)\sigma_\theta(x) + l_\theta(x)\overline{\sigma_\theta}(x)\right)  \otimes \operatorname{Re}[\sigma_\theta(x)]\Big] \\
    & = \frac{1}{2}\mathbb{E} \Big[ \overline{l_\theta}(x) \sigma_\theta(x)  \otimes \operatorname{Re}[\sigma_\theta(x)] + l_\theta(x)\overline{\sigma_\theta}(x)  \otimes \operatorname{Re}[\sigma_\theta(x)]\Big] \\
    & \approx \frac{1}{2}\mathbb{E} \big[ \overline{l_\theta}(x)\big] \mathbb{E}\Big[\sigma_\theta(x)  \otimes \operatorname{Re}[\sigma_\theta(x)]\Big] + \frac{1}{2} \mathbb{E}\big[ l_\theta(x)\big]\mathbb{E}\Big[\overline{\sigma_\theta}(x)  \otimes \operatorname{Re}[\sigma_\theta(x)]\Big] \\
    & = \frac{1}{2} \mathbb{E}\big[ l_\theta(x)\big]\mathbb{E}\big[\big(\sigma_\theta(x) + \overline{\sigma_\theta}(x) \big)  \otimes \operatorname{Re}[\sigma_\theta(x)]\Big] \\
    & = \mathbb{E}\big[l_\theta(x)\big]\mathbb{E} \Big[ \operatorname{Re}[\sigma_\theta(x)]  \otimes \operatorname{Re}[\sigma_\theta(x)]\Big]
\end{align}
where in the third equality we have assumed that $l_\theta$ and $\sigma_\theta$ are approximately independent under $\rho_\theta$. Plugging back into the right-hand side of \eqref{e:optimalbaseline} we find the approximate solution $B \approx \mathbb{E}[l_\theta(x)] \mathbbm{1}$.

\section{Quadratic Hamiltonian in one dimension}\label{app:free}
In one problem dimension ($d=1$) and in the limit of vanishing interaction potential,
\begin{equation}
    H = \frac{1}{2} 
    \begin{bmatrix}
    x \\ p
    \end{bmatrix}
    \sbullet
    \begin{bmatrix}
    h_{xx} & h_{xp} \\
    h_{xp} & h_{pp}
    \end{bmatrix}
    \begin{bmatrix}
    x \\
    p
    \end{bmatrix}
\end{equation}
where $h_{xx} h_{pp}-(h_{xp})^2 > 0$
and the ground eigenfunction is of the form \eqref{e:d=1} with
\begin{align}
    a & = \frac{\sqrt{h_{xx} h_{pp}-(h_{xp})^2}}{h_{pp}} \\
    b & = \frac{h_{xp}}{h_{pp}}
\end{align}
\section{Derivation of time-dependent variational principles}\label{app:tdvp}
In this section we review the derivation of \eqref{e:tdvp}, synthesizing the results of \cite{stokes2020quantum, barison2021efficient}. In the following summation over repeated indices is implied. Denote $\widetilde{\psi}_\theta := e^{-A \delta t} \psi_\theta$. Then
\begin{equation}\label{e:projected}
    \argmin_{\delta\theta \in \mathbb{R}^n} d_{\rm FS}(P_{\widetilde{\psi}_\theta}, P_{\psi_{\theta+\delta\theta}})
    = \argmax_{\delta\theta \in \mathbb{R}^n}  \frac{\left| \left\langle \widetilde{\psi}_\theta | \psi_{\theta + \delta\theta} \right\rangle \right|}{|\langle \widetilde{\psi}_\theta | \widetilde{\psi}_\theta \rangle \langle \psi_{\theta+\delta\theta} | \psi_{\theta+\delta\theta} \rangle|}
	 = \argmax_{\delta\theta \in \mathbb{R}^n}  \left| \left\langle \widetilde{\psi}_\theta | \psi_{\theta + \delta\theta} \right\rangle \right|^2  \enspace ,
\end{equation}
where we used the monotonicity of elementary functions and the normalization of $\psi_{\theta + \delta\theta}$. We have
\begin{align}
	\langle \widetilde{\psi}_\theta |\psi_{\theta + \delta \theta} \rangle
		& =
		\langle \widetilde{\psi}_\theta | \psi_{\theta} \rangle + \left\langle  \widetilde{\psi}_\theta \left| \frac{\partial \psi_\theta}{\partial \theta^{\mu}} \right\rangle\right.\delta \theta^{\mu} 
		+ \frac{1}{2}\left\langle \widetilde{\psi}_\theta \left| \frac{\partial^2 \psi_\theta}{\partial \theta^{\mu} \partial \theta^{\nu}} \right\rangle \right. \delta \theta^{\mu} \delta \theta^{\nu} + \cdots \enspace .
\end{align}
So Taylor expanding $|\langle \widetilde{\psi}_\theta | \psi_{\theta + \delta \theta} \rangle|^2$ to quadratic order in the displacement gives,
\begin{align}
	|\langle \widetilde{\psi}_\theta | \psi_{\theta + \delta \theta} \rangle|^2
	& = |\langle \widetilde{\psi}_\theta | \psi_{\theta} \rangle|^2 +
		\left[\langle \psi_\theta | \widetilde{\psi}_\theta \rangle \left\langle \widetilde{\psi}_\theta \left| \frac{\partial \psi_\theta}{\partial \theta^{\mu}} \right\rangle \right. + \left.\left\langle \frac{\partial \psi_\theta}{\partial \theta^{\mu}} \right| \widetilde{\psi}_\theta \right\rangle \langle \widetilde{\psi}_\theta | \psi_\theta \rangle \right]\delta \theta^{\mu}  +  \\
	& \quad + \left[ \left.\left\langle \frac{\partial \psi_\theta}{\partial \theta^{\mu}} \right| \widetilde{\psi}_\theta \right\rangle \left\langle \widetilde{\psi}_\theta \left| \frac{\partial \psi_\theta}{\partial \theta^{\nu}} \right\rangle\right. +  \frac{1}{2} \langle \psi_\theta | \widetilde{\psi}_\theta \rangle \left\langle \widetilde{\psi}_\theta \left| \frac{\partial^2 \psi_\theta}{\partial \theta^{\mu} \partial \theta^{\nu}} \right\rangle\right. +  \frac{1}{2}\left.\left\langle \frac{\partial^2 \psi_\theta}{\partial \theta^{\mu} \theta^{\nu}} \right| \widetilde{\psi}_\theta \right\rangle \langle \widetilde{\psi}_\theta | \psi_\theta \rangle  \right]\delta \theta^{\mu} \delta \theta^{\nu} + \cdots \enspace  \notag .
\end{align}
Expanding the exponential $e^{-A \delta t}$ in $\delta t$ and neglecting cubic-order terms in the multi-variable Taylor expansion in $\delta\theta$ and $\delta t$,
\begin{equation}
	|\langle \widetilde{\psi}_\theta | \psi_{\theta + \delta \theta} \rangle|^2 =|\langle \widetilde{\psi}_\theta | \psi_{\theta} \rangle|^2  -2\operatorname{Re}\left[\left.\left\langle \frac{\partial \psi_\theta}{\partial \theta^{\mu}}\right| A \psi_\theta \right\rangle + \langle A\psi_\theta | \psi_\theta \rangle \left.\left\langle \frac{\partial \psi_\theta}{\partial \theta^{\mu}} \right| \psi_\theta \right\rangle \right] \delta \theta^{\mu} \delta t - \operatorname{Re}[G_{\mu\nu}(\theta)]\delta \theta^{\mu} \delta\theta^{\nu} + \cdots \enspace ,
\end{equation}
The first-order optimality condition $0=\frac{\partial}{\partial \delta\theta^{\mu}}|\langle \widetilde{\psi}_\theta | \psi_{\theta + \delta \theta} \rangle|^2$, at lowest order in $\delta\theta$ and $\delta t$, thus gives
\begin{align}
	0
		& =  -\operatorname{Re}\left[\left.\left\langle \frac{\partial \psi_\theta}{\partial \theta^{\mu}} \right| A \psi_\theta \right\rangle + \langle A\psi_\theta | \psi_\theta \rangle \left.\left\langle \frac{\partial \psi_\theta}{\partial \theta^{\mu}} \right| \psi_\theta \right\rangle \right] \delta t - \operatorname{Re}[G_{\mu\nu}(\theta)] \delta \theta^{\nu} + \cdots \enspace .
\end{align}
In the limit $\delta t \to 0$, neglecting higher order terms gives \eqref{e:tdvp}.

\section{Field theory motivation}\label{app:field}
In this section we use the heuristic arguments to motivate a Hamiltonian of the form \eqref{e:ham}. Consider the Hamiltonian for a neutral scalar field with quartic self-coupling defined, for simplicity, on the compact interval\footnote{Compactness is convenient because it avoids the possibility of spontaneous symmetry breakdown.} $I=[0,L]$ with periodic boundary conditions
\begin{equation}
    H  = \int_{I} {\rm d}x \, h(x) \, , \quad \quad \quad h(x) = \frac{1}{2} \big[ \pi(x)^2 + \phi'(x)^2 + \mu\, \phi(x)^2 \big] + \frac{\lambda}{4!}\phi(x)^4
\end{equation}
where the field amplitude and the momentum density satisfy $[\phi(x), \pi(y)]=i\delta(x-y)$ for $x,y\in I$.
If we now expand the field operators $\phi$ and $\pi$ using an orthonormal system of periodic real-valued functions $( f_i )_{i \in \mathbb{N}}$,
\begin{equation}
    \phi(x) = \sum_{i\in\mathbb{N}} \hat \phi_i f_i(x) \enspace , \quad \quad \quad \pi(x) = \sum_{i\in\mathbb{N}} \hat \pi_i f_i(x)
\end{equation}
then we find that the operator-valued coefficients $\hat \phi_i$ and $\hat \pi_i$ satisfy the standard canonical commutation relations $[\hat \phi_i, \hat \pi_j]=i\delta_{ij}$. Substituting back into the Hamiltonian, using orthonormality and truncating the infinite series using the $d$th partial sums, we obtain a regulated Hamiltonian which is represented on $L^2(\mathbb{R}^d)$ by the Schr\"{o}dinger operator \eqref{e:ham} with the following identifications 
\begin{equation}
    h =
    \begin{bmatrix}
    \mu\mathbbm{1} + \alpha & 0 \\
    0 & \mathbbm{1}
    \end{bmatrix}
    , \quad\quad \alpha_{ij} = \int_{I} {\rm d} x \, f_i'(x)f_j'(x), \quad\quad \lambda_{ijkl} = \lambda \int_{I} {\rm d} x \, f_i(x)f_j(x)f_k(x)f_l(x)
\end{equation}
Repeating the above analysis for a charged scalar field with nonzero chemical potential we obtain a similiar Hamiltonian in which the off-diagonal blocks of $h$ are nonzero.

\section{Technical proofs}
\begin{proof}[Proof of \eqref{e:grouppcompat}]
Let $f,g \in \Diff(\mathbb{R}^d)$ and $\psi \in L^p(\mathbb{R}^d)$. Then
\begin{align}
    (f \circ g) \cdot \psi(x)
        & = \left|\det J_{g^{-1}\circ f^{-1}}(x)\right|^{1/p} \psi\big(g^{-1}\circ f^{-1}(x)\big) \\
        & = \left|\det J_{g^{-1}}(f^{-1}(x)) \det J_{f^{-1}}(x)\right|^{1/p} \psi\big(g^{-1}\circ f^{-1}(x)\big) \\
        & =
        \left|
        \det J_{f^{-1}}(x)\right|^{1/p}\left|\det J_{g^{-1}}(f^{-1}(x))\right|^{1/p} \psi\big(g^{-1}\circ f^{-1}(x)\big) \\
        & = f \cdot (g \cdot \psi)(x)
\end{align}
\end{proof}

\begin{proof}[Proof of \eqref{e:unitarity}]
Suppose that $f \in \Diff(\mathbb{R}^d)$ and let $y = f(x)$ denote the image of $x \in \mathbb{R}^d$. By the inverse function theorem,
$
    J_{f^{-1}}(y) = J_f(x)^{-1}
$
so $\left|\det J_{f^{-1}}(y)\right| = \left|\det J_f(x)\right|^{-1}$.
Now ($p=2$ here)
\begin{align}
    \left\langle \psi \middle| f \cdot \psi' \right\rangle
        & = \int_{\mathbb{R}^d} {\rm d}y \, \overline{\psi(y)} \left|\det J_{f^{-1}}(y) \right|^{1/2} \psi'\big(f^{-1}(y)\big) \\
        & = \int_{\mathbb{R}^d} {\rm d}x \left|\det J_f(x)\right| \overline{\psi(f(x))}  \left|\det J_{f}(x)\right|^{-1/2} \psi'(x)\\
        & = \int_{\mathbb{R}^d} {\rm d}y \, \overline{ \left|\det J_f(x)\right|^{1/2} \psi(f(x))} \psi'(x) \\
        & = \left\langle f^{-1}\cdot \psi \middle| \psi' \right\rangle \enspace .
\end{align}
Now by \eqref{e:grouppcompat},
\begin{align}
    U_f (U_{f^{-1}}(\psi))
        = f \cdot (f^{-1} \cdot \psi)
        = (f \circ f^{-1})(\psi)
        = \psi
\end{align}
and likewise $U_{f^{-1}} (U_f(\psi)) = \psi$ so $U_{f^{-1}}=U_f^{-1}$.
\end{proof}

\begin{proof}[Proof of \eqref{e:FIMdiffinvariance}]
Recalling the definition \eqref{e:trans} of $f\cdot p_\theta$ (for $p=1$) and using the fact that $f \in \Diff(\mathbb{R}^d)$ is independent of $\theta$,
\begin{align}
    I[f\cdot p_\theta]
    & = \int_{\mathbb{R}^d} {\rm d}x \, (f\cdot p_\theta)(x) \left[\nabla_\theta \log  (f\cdot p_\theta)(x) \, \nabla_\theta \log (f\cdot p_\theta)(x)^T\right] \\
    & = \int_{\mathbb{R}^d} {\rm d}x \, \left| \det J_{f^{-1}}(x) \right| p_\theta(f^{-1}(x))\left[\nabla_\theta \log  p_\theta(f^{-1}(x)) \, \nabla_\theta \log p_\theta(f^{-1}(x))^T\right] \\
    & = I[p_\theta]
\end{align}
\end{proof}

\begin{proof}[Proof of \eqref{e:covariance}]
By the chain rule
\begin{align}
    \frac{\partial}{\partial \theta^\mu} \log \widetilde p_\theta(x) 
    & = \sum_{\nu=1}^n \frac{\partial \phi^\nu(\theta)}{\partial \theta^\mu}\left.\frac{\partial}{\partial \theta^\nu}\right|_{\phi(\theta)} \log p_\theta(x) \\
    & = \sum_{\nu=1}^n (J_\phi)^\nu_{\phantom{\nu}\mu}(\theta)\left.\frac{\partial}{\partial \theta^\nu}\right|_{\phi(\theta)} \log p_\theta(x)
\end{align}
and thus
\begin{align}
    \widetilde I[\theta]
    & := \underset{x \sim \widetilde p_\theta}{\mathbb{E}}\left[\nabla_\theta \log \widetilde p_\theta(x) \, \nabla_\theta \log \widetilde p_\theta(x)^T\right] \\
    & = J_\phi(\theta)^T I(\phi(\theta)) J_\phi(\theta)
\end{align}
\end{proof}

\begin{proof}[Proof of \eqref{e:phaseinv}]
Let $\widetilde{\psi}_\theta := e^{i \omega(\theta)} \psi_\theta$. Then
\begin{align}
    \left\langle \frac{\partial\widetilde{\psi}_\theta}{\partial \theta^\mu}  \middle| \frac{\partial\widetilde{\psi}_\theta}{\partial \theta^\nu}  \right\rangle
    & = \left\langle \frac{\partial\psi_\theta}{\partial \theta^\mu}  + i \frac{\partial \omega}{\partial \theta^\mu}\psi_\theta \middle| \frac{\partial\psi_\theta}{\partial \theta^\nu}  + i \frac{\partial \omega}{\partial \theta^\nu} \psi_\theta \right\rangle \\
    & = \left \langle \frac{\partial\psi_\theta}{\partial \theta^\mu}  \middle| \frac{\partial\psi_\theta}{\partial \theta^\nu}  \right\rangle + \frac{\partial \omega}{\partial \theta^\mu} \frac{\partial \omega}{\partial \theta^\nu} - i \frac{\partial \omega}{\partial \theta^\mu} \left\langle \psi_\theta \middle| \frac{\partial\psi_\theta}{\partial \theta^\nu}  \right\rangle + i \frac{\partial \omega}{\partial \theta^\nu} \left\langle \frac{\partial\psi_\theta}{\partial \theta^\mu}  \middle| \psi_\theta \right\rangle
\end{align}
and
\begin{align}
    \left\langle \frac{\partial\widetilde{\psi}_\theta}{\partial \theta^\mu}   \middle| \widetilde{\psi}_\theta \right\rangle
    & = \left\langle \frac{\partial\psi_\theta}{\partial \theta^\mu}  + i \frac{\partial \omega}{\partial \theta^\mu}\psi_\theta \middle| \psi_\theta \right\rangle \\
    & = \left\langle \frac{\partial\psi_\theta}{\partial \theta^\mu}  \middle| \psi_\theta \right\rangle - i \frac{\partial \omega}{\partial \theta^\mu}
 \end{align}
so
\begin{align}
    \left\langle \frac{\partial\widetilde{\psi}_\theta}{\partial \theta^\mu}   \middle| \widetilde{\psi}_\theta \right\rangle
    \left\langle \widetilde{\psi}_\theta \middle| \frac{\partial\widetilde{\psi}_\theta}{\partial \theta^\nu}  \right\rangle
    & = \left\langle \frac{\partial\psi_\theta}{\partial \theta^\mu}  \middle| \psi_\theta \right\rangle \left\langle \psi_\theta \middle| \frac{\partial\psi_\theta}{\partial \theta^\nu} \right\rangle + \frac{\partial \omega}{\partial \theta^\mu} \frac{\partial \omega}{\partial \theta^\nu} - i \frac{\partial \omega}{\partial \theta^\mu} \left\langle \psi_\theta \middle| \frac{\partial\psi_\theta}{\partial \theta^\nu}  \right\rangle + i \frac{\partial \omega}{\partial \theta^\nu} \left\langle \frac{\partial\psi_\theta}{\partial \theta^\mu}  \middle| \psi_\theta \right\rangle
\end{align}
So
\begin{equation}
    \left\langle \frac{\partial\widetilde{\psi}_\theta}{\partial \theta^\mu}   \middle| \frac{\partial\widetilde{\psi}_\theta}{\partial \theta^\nu} \right\rangle - \left\langle \frac{\partial\widetilde{\psi}_\theta}{\partial \theta^\mu}   \middle| \widetilde{\psi}_\theta \right\rangle
    \left\langle \widetilde{\psi}_\theta \middle| \frac{\partial\widetilde{\psi}_\theta}{\partial \theta^\nu}  \right\rangle
    =
    \left\langle \frac{\partial\psi_\theta}{\partial \theta^\mu}   \middle| \frac{\partial\psi_\theta}{\partial\theta^\nu} \right\rangle - \left\langle \frac{\partial\psi_\theta}{\partial \theta^\mu}   \middle| \psi_\theta \right\rangle
    \left\langle \psi_\theta \middle| \frac{\partial\psi_\theta}{\partial \theta^\nu}  \right\rangle
\end{equation}
\end{proof}

\begin{proof}[Proof of \eqref{e:qgt_unnormalized} and \eqref{e:force_unnormalized}]
Details are provided only for \eqref{e:qgt_unnormalized} since \eqref{e:force_unnormalized} follows similarly. By the chain rule
\begin{align}
    \frac{\partial \psi_\theta}{\partial \theta^\mu}
    & = \frac{1}{\langle \Psi_\theta | \Psi_\theta \rangle^{1/2}}\left[ \frac{\partial \Psi_\theta}{\partial \theta^\mu} - \frac{1}{2} \frac{\Psi_\theta}{\langle \Psi_\theta | \Psi_\theta \rangle} \frac{\partial \langle \Psi_\theta | \Psi_\theta \rangle}{\partial \theta^\mu}\right]
\end{align}
So
\begin{align}
\left\langle 
\frac{\partial\psi_\theta}{\partial \theta^\mu}   
\middle| 
\frac{\partial\psi_\theta}{\partial \theta^\nu} \right\rangle
& =
\frac{1}{\langle \Psi_\theta | \Psi_\theta \rangle}\left\langle 
 \frac{\partial \Psi_\theta}{\partial \theta^\mu} - \frac{1}{2} \frac{\Psi_\theta}{\langle \Psi_\theta | \Psi_\theta \rangle} \frac{\partial \langle \Psi_\theta | \Psi_\theta \rangle}{\partial \theta^\mu}
\middle| 
\frac{\partial \Psi_\theta}{\partial \theta^\nu} - \frac{1}{2} \frac{\Psi_\theta}{\langle \Psi_\theta | \Psi_\theta \rangle} \frac{\partial \langle \Psi_\theta | \Psi_\theta \rangle}{\partial \theta^\nu}
\right\rangle \\
& =
\frac{1}{\langle \Psi_\theta | \Psi_\theta \rangle}
\left[
\left\langle 
\frac{\partial\Psi_\theta}{\partial \theta^\mu}   
\middle| 
\frac{\partial\Psi_\theta}{\partial \theta^\nu} \right\rangle
+
\frac{1}{4}
\frac{1}{\langle \Psi_\theta | \Psi_\theta \rangle}
\frac{\partial \langle \Psi_\theta | \Psi_\theta \rangle}{\partial \theta^\mu}
\frac{\partial \langle \Psi_\theta | \Psi_\theta \rangle}{\partial \theta^\nu} + \right. \\
& \quad
\left.
-
\frac{1}{2}
\frac{1}{\langle \Psi_\theta | \Psi_\theta \rangle}
\left\langle
\frac{\partial \Psi_\theta}{\partial \theta^\mu}
\middle|
\Psi_\theta
\right\rangle
\frac{\partial \langle \Psi_\theta | \Psi_\theta \rangle}{\partial \theta^\nu}
-
\frac{1}{2}
\frac{1}{\langle \Psi_\theta | \Psi_\theta \rangle}
\left\langle
\Psi_\theta
\middle|
\frac{\partial \Psi_\theta}{\partial \theta^\nu}
\right\rangle
\frac{\partial \langle \Psi_\theta | \Psi_\theta \rangle}{\partial \theta^\mu}
\right]
\end{align}
and
\begin{align}
    \left\langle \frac{\partial\psi_\theta}{\partial \theta^\mu}   \middle| \psi_\theta \right\rangle
    \left\langle \psi_\theta \middle| \frac{\partial\psi_\theta}{\partial \theta^\nu}  \right\rangle
    & =
\frac{1}{\langle \Psi_\theta | \Psi_\theta \rangle^2}
\left[\left\langle\frac{\partial \Psi_\theta}{\partial \theta^\mu} \middle| \Psi_\theta \right\rangle- \frac{1}{2} \frac{\partial \langle \Psi_\theta | \Psi_\theta \rangle}{\partial \theta^\mu}\right]
\left[\left\langle\Psi_\theta \middle|\frac{\partial \Psi_\theta}{\partial \theta^\nu} \right\rangle- \frac{1}{2} \frac{\partial \langle \Psi_\theta | \Psi_\theta \rangle}{\partial \theta^\nu}\right]
\end{align}
Expanding out one finds that the offending terms cancel and we obtain
\begin{equation}
    \left\langle \frac{\partial\psi_\theta}{\partial \theta^\mu}   \middle| \frac{\partial\psi_\theta}{\partial\theta^\nu} \right\rangle - \left\langle \frac{\partial\psi_\theta}{\partial \theta^\mu}   \middle| \psi_\theta \right\rangle
    \left\langle \psi_\theta \middle| \frac{\partial\psi_\theta}{\partial \theta^\nu}  \right\rangle 
    =
    \frac{1}{\langle \Psi_{\theta} | \Psi_{\theta} \rangle}\left.\left\langle \frac{\partial \Psi_{\theta}}{\partial \theta^{\mu}}  \right| \frac{\partial \Psi_{\theta}}{\partial \theta^{\nu}} \right\rangle -
	\frac{1}{\langle \Psi_{\theta} | \Psi_{\theta} \rangle^2}\left.\left\langle \frac{\partial \Psi_{\theta}}{\partial \theta^{\mu}} \right| \Psi_{\theta}\right\rangle  \left\langle \Psi_{\theta} \left| \frac{\partial \Psi_{\theta}}{\partial \theta^{\nu}}  \right\rangle\right.
\end{equation}
\end{proof}

\begin{proof}[Proof of \eqref{e:holotdvp}]
Using \eqref{e:holoG} and \eqref{e:holoforce} we obtain the following $\theta = \theta_1 \oplus \theta_2$ decompositions,
\begin{align}
    g(\theta)\dot{\theta}
    &=
    \begin{bmatrix}
    \operatorname{Re}[S(z)] \dot{\theta}_1 - \operatorname{Im}[S(z)]\dot{\theta}_2\\
    \operatorname{Im}[S(z)] \dot{\theta}_1 + \operatorname{Re}[S(z)]\dot{\theta}_2
    \end{bmatrix} 
    \\
    -\operatorname{Re}[F(\theta)] 
    &=
    \begin{bmatrix}
    -\operatorname{Re}[F(z)] \\
    -\operatorname{Im}[F(z)]
    \end{bmatrix}
    \\
    \operatorname{Im}[F(\theta)] 
    &=
    \begin{bmatrix}
    \operatorname{Im}[F(z)] \\
    -\operatorname{Re}[F(z)]
    \end{bmatrix}
\end{align}
Plug these expressions into the evolution equations \eqref{e:tdvp} and consider a linear combination consisting of the $\theta_1$ rows superposed with an imaginary unit multiplying the $\theta_2$ rows to obtain \eqref{e:holotdvp}.
\end{proof}

\begin{proof}[Proof of \eqref{e:holorate}]
Taking the time derivative of the loss function $\mathcal{
L}$ defined in \eqref{e:gradientflow} and assuming that the holomorphic constraints \eqref{e:holo} are satisfied we obtain,
\begin{align}
    \dot{\mathcal{L}}(\theta)
    & = \dot{\theta}_1^T \nabla_{\theta_1} \mathcal{L}(\theta) + \dot{\theta}_2^T \nabla_{\theta_2} \mathcal{L}(\theta) \\
    & = \dot{\theta}_1^T \operatorname{Re}[F(z)] + \dot{\theta}_2^T \operatorname{Im}[F(z)] \\
    & = \operatorname{Re}[\dot{z}^\dag F(z)] \\
    & = 
    \begin{cases}
    -\operatorname{Re}[\dot{z}^\dag S(z) \dot{z}] \\
    \operatorname{Re}[i \dot{z}^\dag S(z) \dot{z}]
    \end{cases} \\
    & = 
    \begin{cases}
    -\dot{z}^\dag S(z) \dot{z} \\
    0
    \end{cases}
\end{align}
The first equality is the chain rule. The second equality used $\nabla \mathcal{L}(\theta) = \operatorname{Re}[F(\theta)]$ together with \eqref{e:holoforce}. The third equality used $\dot{z}=\dot{\theta}_1 + i\dot{\theta}_2$. The fourth equality used \eqref{e:holotdvp} and the final equality used the fact that $S(z) \in \mathbb{H}^m_+$ is Hermitian positive semi-definite.
\end{proof}

\begin{proof}[Proof of \eqref{e:gF}]
Recalling our conventions for complex covariance matrices in section \ref{sec:notation},
\begin{align}
    \cov(\sigma_\theta,\sigma_\theta)
    & = \mathbb{E}[\sigma_\theta(x) \sigma_\theta(x)^\dag] - \mathbb{E}[\sigma_\theta(x)] \mathbb{E}[\sigma_\theta(x)]^\dag \\
    \cov(\sigma_\theta,\sigma_\theta)^T
    & = \mathbb{E}[\overline{\sigma_\theta}(x) \sigma_\theta(x)^T] - \mathbb{E}[\overline{\sigma_\theta}(x)] \mathbb{E}[\sigma_\theta(x)]^T \\
    & = G(\theta)
\end{align}
Similarly,
\begin{align}
    \cov(l_\theta,\sigma_\theta)
    & = \mathbb{E}[l_\theta(x) \sigma_\theta(x)^\dag] - \mathbb{E}[l_\theta(x)] \mathbb{E}[\sigma_\theta(x)]^\dag \\
    \cov(l_\theta,\sigma_\theta)^T
    & = \mathbb{E}[\overline{\sigma_\theta}(x) l_\theta(x)] - \mathbb{E}[\overline{\sigma_\theta}(x)] \mathbb{E}[l_\theta(x)]
\end{align}
\end{proof}

\begin{proof}[Proof of \eqref{e:canongrad}]
Starting from \eqref{e:stochoptim} we obtain,
\begin{align}
    \nabla \mathcal{L}(\theta)
    & = \mathbb{E}\left[
    \left(\hat{\mathcal{L}}_\theta(x)\mathbbm{1} - \frac{B}{2}\right)\nabla_\theta \log \rho_\theta(x) \right]
    +
    \mathbb{E}\left[\nabla_\theta
    \hat{\mathcal{L}}_\theta(x)\right] \\
    & = \frac{1}{2}\mathbb{E}\left[
    \left(l_\theta(x)\mathbbm{1} - B\right)\big(\overline{\sigma_\theta}(x) + \sigma_\theta(x)\big) \right]
    +
    \frac{1}{2}\mathbb{E}\left[\nabla_\theta
    l_\theta(x)\right]
\end{align}
where we have used 
\begin{align}
    \nabla_\theta \log \rho_\theta(x)
    & = \frac{\nabla_\theta \rho_\theta(x)}{\rho_\theta(x)} \\
    & = \frac{\nabla_\theta \overline{\psi_\theta}(x) \, \psi_\theta(x) + \overline{\psi_\theta}(x)\nabla_\theta \psi_\theta(x)}{|\psi_\theta(x)|^2} \\
    & = \overline{\sigma_\theta}(x) + \sigma_\theta(x)
\end{align}
Now
\begin{align}
    \mathbb{E}\left[\nabla_\theta
    l_\theta(x)\right]
    & = \int_{\mathbb{R}^d} {\rm d}x \, | \psi_\theta(x)|^2 \left[\frac{\nabla_\theta (H\psi_\theta)(x)}{\psi_\theta(x)} - \frac{(H\psi_\theta)(x)}{\psi_\theta(x)}\frac{\nabla_\theta \psi_\theta(x)}{\psi_\theta(x)}\right] \\
    & = \int_{\mathbb{R}^d} {\rm d}x \, \overline{\psi_\theta}(x) \nabla_\theta (H\psi_\theta)(x) -\mathbb{E}\big[l_\theta(x) \sigma_\theta(x)\big] \\
    & = \int_{\mathbb{R}^d} {\rm d}x \, \overline{(H\psi_\theta)}(x) \nabla_\theta \psi_\theta(x) -\mathbb{E}\big[l_\theta(x) \sigma_\theta(x)\big] \\
    & = \int_{\mathbb{R}^d} {\rm d}x \, |\psi_\theta(x)|^2 \overline{\left[\frac{(H\psi_\theta)(x)}{\psi_\theta(x)}\right]}\frac{\nabla_\theta \psi_\theta(x)}{\psi_\theta(x)}  -\mathbb{E}\big[l_\theta(x) \sigma_\theta(x)\big] \\
    & = \mathbb{E}\big[\overline{l_\theta}(x)\sigma_\theta(x)-l_\theta(x)\sigma_\theta(x)\big]
\end{align}
where we have interchanged the order of operations of the gradient $\nabla_\theta$ with the Hamiltonian $H$ and also used the fact that $H$ is Hermitian. Thus,
\begin{align}
    \nabla \mathcal{L}(\theta)
    & = \frac{1}{2} \mathbb{E}\big[l_\theta(x)\overline{\sigma_\theta}(x) + \overline{l_\theta}(x)\sigma_\theta(x)\big] - \frac{1}{2} B \, \mathbb{E}\big[\sigma_\theta(x) + \overline{\sigma_\theta}(x)\big] \\
    & = \operatorname{Re}\mathbb{E}\big[\big(l_\theta(x) - B\big)\overline{\sigma_\theta}(x)\big]
\end{align}
Alternatively, starting from the definition of the loss function \eqref{e:gradientflow},
\begin{align}
    \nabla \mathcal{L}(\theta) 
        & = \frac{\langle \nabla_\theta \psi_\theta | H \psi_\theta\rangle + \langle \psi_\theta | H \nabla_\theta \psi_\theta\rangle}{2}   \\
        & = \frac{\langle \nabla_\theta \psi_\theta | H \psi_\theta\rangle + \langle H\psi_\theta | \nabla_\theta \psi_\theta\rangle}{2} \\
        & = \frac{1}{2}\langle \nabla_\theta \psi_\theta | H\psi_\theta \rangle
        +
        \textrm{c.c.}
        \\
        & =
        \frac{1}{2}{\mathbb{E}} \big[  l_\theta(x) \overline{\sigma_\theta}(x) \big] + \textrm{c.c.} \\
        & =
        {\mathbb{E}} \operatorname{Re}\big[ l_\theta(x) \overline{\sigma_\theta}(x) \big]
\end{align}
where we have used the product rule in the first equality, Hermiticity of $H$ in the second equality and conjugate symmetry of $\langle \cdot | \cdot \rangle$ in the third equality.
\end{proof}

\begin{proof}[Proof of \eqref{e:levelsets}]
By definition of the zero level set,
\begin{align}
    L_0(\operatorname{mod} f\cdot \psi)
        & := \{ x \in \mathbb{R}^d : | \det J_{f^{-1}}(x)|^{1/p}\operatorname{mod} \psi\big(f^{-1}(x)\big) = 0 \} \\
        & = \{ x \in \mathbb{R}^d : \operatorname{mod} \psi\big(f^{-1}(x)\big) = 0 \} \\
        & =: (\operatorname{mod}\psi \circ f^{-1})^{-1}[\{0\}]\\
        & = f \big((\operatorname{mod}\psi)^{-1}[\{0\}]\big) \\
        & = f \big(L_0(\operatorname{mod}\psi)\big)
\end{align}
where the second equality used the fact that $f$ is a diffeomorphism to divide out the everywhere nonzero Jacobian factor $|\det J_{f^{-1}}(x)| > 0$, the third equality is by definition of the pre-image and the fourth equality follows from the property of pre-images under a composition of maps. By the same reasoning,
\begin{align}
    L_\theta(\arg f\cdot \psi)
        & = \{ x \in \mathbb{R}^d : \arg\psi\big(f^{-1}(x)\big) = \theta \} \\
        & =: (\arg \psi \circ f^{-1})^{-1}[\{\theta\}]\\
        & = f\big((\arg \psi)^{-1}[\{\theta\}]\big)\\
        & = f \big(L_\theta(\arg\psi)\big)
\end{align}
\end{proof}

\begin{proof}[Proof of Lemma \ref{lem:rep}]
Starting with the expression $\psi(gx)=\rho(g)\psi(x)$, taking the complex modulus, raising to the $p$th power and integrating we obtain
\begin{equation}
    \int_{\mathbb{R}^d} {\rm d}x \, |\psi(gx)|^p 
    = |\rho(g)|^p\int_{\mathbb{R}^d} {\rm d}x \, |\psi(x)|^p
\end{equation}
Now changing integration variables on the left-hand side,
\begin{equation}
    |{\det (g^{-1}})| \int_{\mathbb{R}^d} {\rm d}x \, |\psi(x)|^p
    = |\rho(g)|^p\int_{\mathbb{R}^d} {\rm d}x \, |\psi(x)|^p
\end{equation}
Recalling that $\Vert \psi \Vert_p \neq 0$ and that $G$ is an orthogonal group we conclude $|\rho(g)| = 1$, as required for a one-dimensional unitary representation. Recalling the definition \eqref{e:trans} we obtain
\begin{equation}
    (f\cdot \psi)(gx)
    = \left|\det J_{f^{-1}}(gx)\right|^{1/p} \psi\big(f^{-1}(gx)\big) \enspace .
\end{equation}
Now $\psi\big(f^{-1}(gx)\big) = \psi\big(gf^{-1}(x)\big)=\rho(g)\psi\big(f^{-1}(x)\big)$. In addition, as shown in \cite{papamakarios2021normalizing}, $\left|\det J_{f^{-1}}(gx)\right|=\left|\det J_{f^{-1}}(x)\right|$ and therefore $(f\cdot \psi)(gx) = \rho(g)\,(f\cdot \psi)(x)$.
\end{proof}

\section{Architecture and Training Details}\label{sec:architecture_and_training_details}

In this section, additional details on the architecture and training procedure used for the experiments are provided.

\subsection{Architecture and Training Procedure} \label{sec: architecture}

The training procedure requires both sampling as well as calculating the log probabilities for a given sample. As such, efficiency in both the forward and backward pass of a normalizing flow is desired, which motivates the choice of RealNVP \cite{realnvp2017} as the architecture for the flow. The architecture is modeled off the nflows package \cite{nflows} and built in PyTorch \cite{pytorch}.

The normalizing flows are trained using natural gradient descent, using the Adam optimizer \cite{kingma2017adam} applied to natural gradient estimates. In the calculation of the Fisher information matrix, a preconditioning term $\gamma$ of 0.1 is added to the diagonal in order to stabilize training. We start with an initial learning rate of $0.01$ - this initial learning rate is decayed using a cosine decay schedule with no warm restarts \cite{loshchilov2017sgdr}. The per-sample gradients necessary for calculating the Fisher information matrix are calculated using the Backpack package \cite{dangel2020backpack}. An ablation study comparing different optimization methods is shown in Table \ref{tab:ablation_sr}.

\begin{table}[h]
    \centering
        \begin{tabular}{cccc}
        \hline
        Dimension &          2  &            5  &          10 \\
        \hline
Natural Gradient, $\gamma=1.0$       &  -3.84±0.53 &  -48.76±4.77 &  -3.24±0.42 \\
Natural Gradient, $\gamma=1.0$, Adam &  \textbf{-4.69±0.02} &  -46.75±6.12 &  \textbf{-7.58±0.18} \\
Natural Gradient, $\gamma=0.1$       &  -4.38±0.29 &  -47.41±5.15 &  -6.78±0.59 \\
Natural Gradient, $\gamma=0.1$, Adam &  -4.69±0.02 &  \textbf{-50.12±3.06} &  -7.52±0.26 \\
Standard Gradient, Adam          &  -4.68±0.02 &  -48.99±4.82 &  -7.46±0.26 \\
Gaussian             &  -4.39±0.39 &  -40.22±7.64 &  -6.92±0.43 \\
        \hline
        \hline
        Dimension &           25  &          50  &          100 \\
        \hline
Natural Gradient, $\gamma$=1.0       &  -35.78±7.49 &  -28.38±3.76 &  -75.85±2.95 \\
Natural Gradient, $\gamma$=1.0, Adam &  -47.87±3.43 &  -39.79±1.39 &  -79.43±2.17 \\
Natural Gradient, $\gamma$=0.1       &  -46.63±2.93 &  \textbf{-41.51±1.20} &  \textbf{-80.11±1.77} \\
Natural Gradient, $\gamma$=0.1, Adam &  \textbf{-48.29±3.46} &  -39.91±1.75 &  -79.66±2.54 \\
Standard Gradient, Adam          &  -45.98±3.82 &  -40.97±1.08 &  -76.51±3.18 \\
Gaussian             &  -42.53±2.77 &  -35.95±1.43 &  -72.83±4.00 \\
\hline
        \end{tabular}
        \vspace{0.1in}
    \caption{\textit{Ablation study on some of the different options for using the natural gradient (with or without the Adam optimizer and changing the preconditioning terms $\gamma$ for the Fisher information matrix). Using the standard gradient with the Adam optimizer is also considered. This ablation study uses the standard normalizing flow approach (no symmetrization or adiabatic retraining). The results shown indicate the mean ending energy across ten runs using different random initializations, with the error bounds equal to two times the standard error.}}
    \label{tab:ablation_sr}
\end{table}

\subsection{Randomly Generating Hamiltonians}\label{sec:generating_u_and_h}

In this subsection, we describe the process for selecting the matrices $h_{xx}$ and $u$ which define the Hamiltonian. The procedure for selecting $u$ for dimension $d$ is as follows:
\begin{itemize}
    \item Select $d$ eigenvalues uniformly randomly from the interval $[0.1, 2]$. Let $\Sigma$ be a $d\times d$ diagonal matrix with these eigenvalues along the diagonal.
    \item Sample a random matrix $U$ from the Haar distribution over the orthogonal group in dimension $d$.
    \item Let $u$ equal $U\Sigma U^T$.
\end{itemize}
The procedure for selecting $h_{xx}$ is identical, except in the final step, $h_{xx}$ is set to be equal to $-U\Sigma U^T$. This ensures that $u$ is positive definite and $h_{xx}$ is negative definite, while maintaining bounds on the condition number of each matrix. For each dimension $d$, we run each approach (normalizing flows, symmetric normalizing flows, the normalizing flows approaches with adiabatic retraining, and Gaussian state approximation) ten times using different random initializations.

\subsection{Adiabatic Retraining and Flow-Distance Regularization} \label{sec: adiabatic_retraining}

Following the recent work of \cite{hackett2021flow}, we investigate two methods for improving the final energy, namely adiabatic retraining and flow-distance regularization. Adiabatic retraining involves varying the parameters of the target objective in such a way that it interpolates between a relatively simple problem to a more complicated problem. To implement this, the quadratic term in \eqref{e:ham} is multiplied by a term $\alpha$, where $\alpha$ ranges from $0$ to $1$ during training. Following \cite{hackett2021flow}, $\alpha$ is exponentially decayed - in other words, we let
\begin{equation}
    \alpha = \frac{e^{-kt}-e^{-k}}{1-e^{-k}},
\end{equation}
where $k$ is a hyperparameter. Unlike \cite{hackett2021flow}, adiabatic retraining is used during the entire training procedure, rather than for a interval in the middle of training.

Hackett et al. \cite{hackett2021flow} also introduce flow-distance regularization, which imposes a penalty on the flow for transforming samples $z$ from the base distribution to significantly different outputs. This penalty is enforced using a $l_2$ norm between samples from the base distribution and the outputs. The penalty is annealed to zero by the end of training. Unlike Hackett et al., we did not find a significant difference when using flow distance regularization.

\subsection{Gaussian State Approximation} \label{sec: gaussian_state_approx}
As a point of comparison for the normalizing flows approach, the energy of an optimal Gaussian wave function \eqref{eq: gaussian_wave_fcn} was estimated by optimizing over the variational parameters using Riemannian gradient descent. Setting $\lambda_{ijkl}=3\delta_{ij}\delta_{kl} u_{ik}$ and $h=\diag(h_{xx},\mathbbm{1})$ in the \eqref{e:ham} and setting $B=0$ in \eqref{eq: gaussian_wave_fcn}  we obtain,
\begin{align}
    \langle \psi_G | H \psi_G \rangle =& \frac{1}{4}\tr(A) + \frac{1}{4} \text{sum}\left(h_{xx}\odot A^{-1}\right) + \frac{1}{32}\left(\diag\left(A^{-1}\right) + 2\mu^{\odot 2}\right)^T u \left(\diag\left(A^{-1}\right) + 2\mu^{\odot 2}\right) + \\
    & \frac{1}{16}\text{sum}\left(u\odot \left(A^{-1}\right)^{\odot 2}\right)+\frac{1}{2}\mu^T\left(h_{xx} + \frac{1}{2} u A^{-1}\right)\mu \notag
\end{align}

\section{Loss Function Estimator}\label{sec:ibp_appendix}

In this section we illustrate the adjoint loss function estimator $\hat{\mathcal{L}}_\theta^{\rm (adj)}(x)$ as an alternative to the canonical estimator $\hat{\mathcal{L}}_\theta^{\rm (can)}(x)$ in the simple example of the harmonic oscillator Hamiltonian \eqref{eq:simple_harmonic_oscillator},
\begin{align}
     \hat{\mathcal{L}}_\theta^{\rm (can)}(x) 
     & := \frac{1}{2}l_\theta(x) \\
     & = \frac{1}{4} \left[- \frac{\psi_\theta''(x)}{\psi_\theta(x)} + x^2\right] \\
    \hat{\mathcal{L}}_\theta^{\rm (adj)}(x)
    & = \frac{1}{4}\left[\left| \frac{\psi_\theta'(x)}{\psi_\theta(x)}\right|^2 + x^2\right] \label{eq:adjoint_estimator}
\end{align}

It should be noted that unlike the canonical estimator, the adjoint estimator does not exhibit the zero-variance property, as described in Section \ref{sec: stochastic_optimization_and_variance_reduction}. For the example of the simple harmonic oscillator \eqref{eq:simple_harmonic_oscillator}, consider the variational class given by the family of wave functions \eqref{e:d=1}, parametrized by $a > 0$ with $b$ fixed to zero ($a=1$ is the ground eigenfunction \eqref{eq:simple_harmonic_oscillator_ground_eigenfcn}). Figure \ref{fig:var_of_gaussian} shows the variance of both estimators as a function of the variational parameter $a$, showing clearly the zero-variance principle of the canonical estimator. Although the adjoint estimator is subject to an irreducible quantum uncertainty, the standard error of the estimate can be reduced to zero at the Monte Carlo rate of $1/\sqrt{N}$ simply by increasing the number of samples $N$.

\begin{figure}
    \centering
    \includegraphics[width=0.8\textwidth]{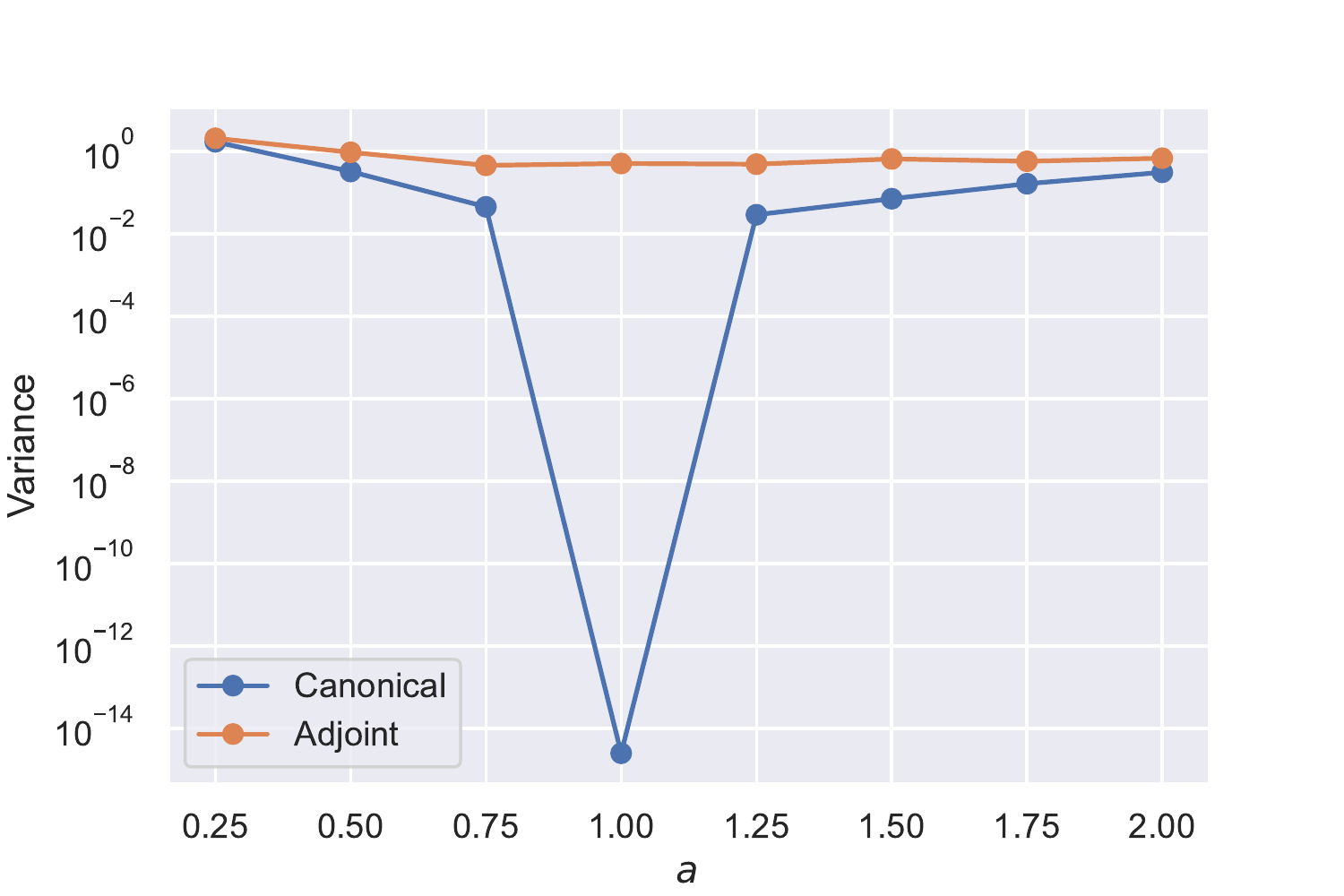}
    \caption{\textit{A plot of the variance of the canonical and adjoint energy estimators $\hat{\mathcal{L}}_\theta^{\rm (can)}(x)$ and $\hat{\mathcal{L}}^{\rm (adj)}_\theta(x)$ for the simple harmonic oscillator Hamiltonian \eqref{eq:simple_harmonic_oscillator} as a function of the variational parameter $a$
    in \eqref{e:d=1} with $b=0$ fixed.}}
    \label{fig:var_of_gaussian}
\end{figure}

\end{appendix}

\nocite{*}
\bibliographystyle{plain}
\bibliography{main}
\end{document}